\documentclass[10pt]{article}

\usepackage[utf8]{inputenc}

\usepackage{hyperref}
\usepackage{verbatim} 

\usepackage{graphicx}
\usepackage{caption}
\usepackage{subcaption}
\usepackage{amssymb}
\usepackage{ textcomp } 
\usepackage{color}
\usepackage{enumerate}
\usepackage{mathrsfs}
\usepackage{multicol}
\usepackage{cite}

\usepackage[ruled,vlined, onelanguage]{algorithm2e}

\usepackage{amsopn}

\usepackage{amsmath}
\usepackage{amsthm}
\usepackage{amssymb}

\usepackage{multicol}

\usepackage[ruled,vlined]{algorithm2e}

\newcommand{\C}{\mathbb{C}}

\newcommand{\N}{\mathbb{N}}

\newcommand{\R}{\mathbb{R}}
\newcommand{\T}{\mathbb{T}}
\newcommand{\Z}{\mathbb{Z}}

\newcommand{\calA}{\mathcal{A}}

\newcommand{\calC}{\mathcal{C}}
\newcommand{\calD}{\mathcal{D}}
\newcommand{\calE}{\mathcal{E}}

\newcommand{\calH}{\mathcal{H}}

\newcommand{\calM}{\mathcal{M}}
\newcommand{\calP}{\mathcal{P}}

\newcommand{\vphi}{\varphi}

\newcommand{\abs}[1]{\vert #1 \vert}
\newcommand{\norm}[1]{\Vert #1 \Vert}
\newcommand{\abss}[1]{\left\vert #1 \right\vert}
\newcommand{\normm}[1]{\left\Vert #1 \right\Vert}
\newcommand{\set}[1]{\left\lbrace #1\right\rbrace}
\newcommand{\sse}{\subseteq}
\newcommand{\sprod}[1]{\left\langle #1 \right\rangle}

\newcommand{\sph}{\mathbb{S}}

\newcommand{\erw}[1]{\mathbb{E}\left( #1 \right)}

\usepackage{dsfont}

\newcommand{\geqsim}{\gtrsim}
\newcommand{\leqsim}{\lesssim}

\newcommand{\emb}{\hookrightarrow}

\newcommand{\hatphi}{\widehat{\phi}}

\DeclareMathOperator{\supp}{supp}

\DeclareMathOperator{\argmin}{argmin}

\DeclareMathOperator{\ran}{ran}

\DeclareMathOperator{\diam}{diam}
\DeclareMathOperator{\re}{Re}

\DeclareMathOperator{\herm}{Herm}

\DeclareMathOperator{\sinc}{sinc}

\newcommand{\st}{\text{ subject to }}

\newtheorem{lem}{Lemma}
\newtheorem{prop}[lem]{Proposition}

\newtheorem{theo}[lem]{Theorem}
\newtheorem{cor}[lem]{Corollary}

\newtheorem{defi}[lem]{Definition}
\newtheorem{rem}[lem]{Remark}

\newtheorem*{mainResult}{Main Result}

\theoremstyle{definition}

\numberwithin{lem}{section}

\usepackage{etoolbox}
\let\bbordermatrix\bordermatrix
\patchcmd{\bbordermatrix}{8.75}{4.75}{}{}
\patchcmd{\bbordermatrix}{\left(}{\left[}{}{}
\patchcmd{\bbordermatrix}{\right)}{\right]}{}{}

\def\mindex#1{\index{#1}}

\def\sq{\hbox{\rlap{$\sqcap$}$\sqcup$}}
\def\qed{\ifmmode\sq\else{\unskip\nobreak\hfil
\penalty50\hskip1em\null\nobreak\hfil\sq
\parfillskip=0pt\finalhyphendemerits=0\endgraf}\fi\medskip}

\long\def\defbox#1{\framebox[.9\hsize][c]{\parbox{.85\hsize}{%
\parindent=0pt
\baselineskip=12pt plus .1pt      
\parskip=6pt plus 1.5pt minus 1pt 
 #1}}}

\long\def\beginbox#1\endbox{\subsection*{}%
\hbox{\hspace{.05\hsize}\defbox{\medskip#1\bigskip}}%
\subsection*{}}

\def\endbox{}

\def\supp{{\rm supp\,}}

\newsavebox{\junk}
\savebox{\junk}[1.6mm]{\hbox{$|\!|\!|$}}

\def\argmin{\mathop{\rm arg\, min}}

\def\bE{{\mathbb E}}

\def\bS{{\mathbb S}}

\def\bfmath#1{{\mathchoice{\mbox{\boldmath$#1$}}%
{\mbox{\boldmath$#1$}}%
{\mbox{\boldmath$\scriptstyle#1$}}%
{\mbox{\boldmath$\scriptscriptstyle#1$}}}}

\def\bfmY{\bfmath{Y}}

\def\bfmhhaY{\bfmath{\hhaY}} 
\def\bfmhhaY{\hbox to 0pt{$\widehat{\bfmY}$\hss}\widehat{\phantom{\raise 1.25pt\hbox{$\bfmY$}}}}

\def\til={{\widetilde =}}

\def\clH{{\cal H}}

\def\clI{{\cal I}}

 \def\FRAC#1#2#3{\genfrac{}{}{}{#1}{#2}{#3}}

\def\ddtp{{\mathchoice{\FRAC{1}{d^{\hbox to 2pt{\rm\tiny +\hss}}}{dt}}%
{\FRAC{1}{d^{\hbox to 2pt{\rm\tiny +\hss}}}{dt}}%
{\FRAC{3}{d^{\hbox to 2pt{\rm\tiny +\hss}}}{dt}}%
{\FRAC{3}{d^{\hbox to 2pt{\rm\tiny +\hss}}}{dt}}}}

\def\average#1,#2,{{1\over #2} \sum_{#1}^{#2}}

\def\eye(#1){{\bf(#1)}\quad}

\def\eq#1/{(\ref{e:#1})}

\newcommand{\inp}[2]{{\langle #1, #2 \rangle}}

\newcommand{\beqn}[1]{\notes{#1}%
\begin{eqnarray} \elabel{#1}}

\newcommand{\eeqn}{\end{eqnarray} }

\newcommand{\beq}[1]{\notes{#1}%
\begin{equation}\elabel{#1}}

\newcommand{\eeq}{\end{equation}}

\def\bdes{\begin{description}}
\def\edes{\end{description}}

\newcounter{rmnum}

\newcounter{anum}

{\end{list}}

\def\ass(#1:#2){(#1\ref{#1:#2})}

\def\ritem#1{
\item[{\sf \ass(\current_model:#1)}]
}

\newenvironment{recall-ass}[1]{%
\begin{description}
\def\current_model{#1}}{
\end{description}
}

\long\def\comment#1{}

\newfont{\bbb}{msbm10 scaled 700}

\newfont{\bb}{msbm10 scaled 1100}

\title{Estimation of Angles of Arrival Through Superresolution -- \\ A
Soft Recovery Approach for General Antenna Geometries}

\usepackage[affil-it]{authblk}

\author[1,2]{Mahdi Barzegar}
\author[1]{Guiseppe Caire}
\author[3]{Axel Flinth}
\author[1]{Saeid Haghighatshoar}
\author[3]{Gitta Kutyniok}
\author[2]{Gerhard Wunder}

\affil[1]{Institut für Telekommunikationssysteme, Technische Universität Berlin}
\affil[2]{Heisenberg Communications and Information Theory Group, Freie Universität Berlin}
\affil[3]{Institut für Mathematik, Technische Universität Berlin}

\begin{document}

\maketitle

%
%
%
%
%
%
%
%
%
%
%
%
%

\begin{abstract}
 The estimation of direction of arrivals with help of $TV$-minimization is studied. Contrary to prior work in this direction, which has only considered certain antenna placement designs, we consider general antenna geometries. Applying the soft-recovery framework, we are able to derive a theoretic guarantee for a certain direction of arrival to be approximately recovered. We discuss the impact of the recovery guarantee for a few concrete antenna designs. Additionally, numerical simulations supporting the findings of the theoretical part are performed.
\end{abstract}

\section{Introduction}

In this paper, we study the estimation of \textit{direction-of-arrivals} (DoAs) of $s$ planar waves from  their superposition received at an array  consisting of $m$ antennas. This problem arises in a variety of applications in radar, communication systems, speech processing, etc. Classical algorithms for DoA estimation  include Pisarenko root finding \cite{pisarenko1973retrieval}, MUSIC \cite{schmidt1986multiple} and ESPRIT \cite{roy1989esprit}. They are sometimes referred to as ``super-resolution'' (SR) methods since they are able to resolve the DoAs in the continuous domain from the observation of the low-dim signal received at  array elements. However, when the minimum angular separation between the planar waves becomes much smaller than the array spatial resolution, the Fisher Information Matrix for the joint estimation of the DoAs tends to be highly ill-conditioned, and all of these methods yield a poor performance \cite{johnson2008music}. 
Another common approach for DoA estimation consists in parametric methods such as maximum likelihood technique, which can be posed as a nonlinear least squares (NLS) optimization. We refer to  \cite{krim1996two, stoica2005spectral} for a more comprehensive overview of  these methods.

By the advent of Compressed Sensing (CS), the DoA estimation has been revisited in the framework of sparsity-based algorithms. The conventional way to cast the DoA estimation as an instance of CS, is to quantize the set of DoAs into a discrete grid. 
Such a grid-based approach has been vastly studied in the compressed sensing literature \cite{bajwa2010compressed, baraniuk2007compressive, duarte2013spectral, fannjiang2010compressed, herman2009high, malioutov2005sparse, kunis2008random, stoica2012spice, stoica2011new}.
However, the assumption that the DoAs belong to the grid leads to some model mismatch, which typically leads to significant performance degradation in DoA estimation \cite{chi2011sensitivity}.
Recently, Cand{\`e}s and Fernandez-Granda \cite{candes2014towards,candes2013super}  reconsidered SR by formulating 
the problem as a convex optimization that considers DoAs in a continuous domain and does not suffer from the mismatch problem of grid-based approaches. 
More specifically, the results proved in \cite{candes2014towards} guarantee the stable recovery of a discrete (complex) measure 
$\mu_0=\sum_{\ell=1}^s w_\ell \delta_{u_\ell}$ over $u \in [0,1)$ from a collection of its low-frequency 
Fourier coefficients 
\begin{align}\label{eq:mu_a}
f_k=(\widehat{\mu_0})_k=\int_{0}^1 e^{i 2k\pi u} d\mu_0(u)=\sum_{\ell=1}^s w_\ell e^{j 2k\pi u_\ell}, k=0,1, \dots, m-1.
\end{align}
In \cite{candes2014towards}, the recovery of the discrete measure $\mu_0$ was cast as  the following convex optimization:
\begin{align}\label{eq:candes_spr}
\mu^*=\argmin_{\mu} \|\mu\|_{TV} \text{ subject to } \widehat{\mu}_k=f_k, k=0,1,\dots, m-1,
\end{align}
where, for a measure $\mu$ over a  domain $\Omega$, the total-variation (TV) norm is defined by $$\norm{\mu}_{TV} = \sup_{ \phi \in \calC_0(\Omega), \norm{\phi}_\infty \leq 1} \int \phi d\mu,$$ where $\calC_0(\Omega)$ denotes the space of continuous functions over $\Omega$ vanishing at infinity\footnote{A function on a topological space $\Omega$ is said to vanish at infinity if for each $\epsilon>0$, there exists a compact set $K$ with $\abs{f}<\epsilon$ outside $K$.}. The TV-norm minimization can be seen as the $\ell_1$-norm minimization over the space of signed (complex) measures and analogous to $\ell_1$-norm minimization in CS \cite{donoho2006compressed,candes2006near} promotes the sparsity of the underlying measure. 

\subsection{TV-Minimization for Estimation of Direction of Arrivals}
TV-minimization for estimation of DoAs has already been studied in  \cite{tan2014direction}. There, only the case of uniform linear arrays are treated. In this paper, we will deal with general antenna geometries with array elements located at arbitrary locations $\{\Delta_\ell\}_{\ell=1}^m\subset \R^2$. This has a practical importance. For instance, new geometries based on non-uniform linear arrays have been shown to yield larger degrees of freedom and better DoA resolution \cite{pal2010nested, vaidyanathan2011sparse}.
This has motivated designing new arrays with 2D and even 3D geometries such as circular and rectangular (lattice) arrays.  In this publication, we will stay in the 2D-regime.

As we will explain thoroughly in the sequel, the estimation of a set of DoAs $(\theta_\ell)_{\ell=1}^s$ can be recast as the recovery problem of a discrete measure $\mu_0 = \sum_{\ell=1}^s c_{\theta_\ell} \delta_{\theta_\ell} \in \calM(\sph^1)$ from possible noisy measurements $b=M\mu_0 +n$, where $M$ is the linear measurement operator defined on $\calM(\sph^1)$ (we will specify $M$ later in the text).  We  make the assumption $\norm{\mu_0}_{TV}=1$, so that $\abs{c_{\theta_\ell}}$ can be interpreted as the \emph{relative} power of the $\ell$:th peak. We  use the following program to recover $\mu_0$ in the noiseless case 
\begin{align}
	\min \norm{\mu}_{TV} \st M\mu=b \tag{$\calP_{TV}$}, \label{eq:PTV}
\end{align}
and the following one in the noisy case.
\begin{align}
	\min \norm{M\mu -b}_2 \st \norm{\mu}_{TV} \leq \rho. \tag{$\calP_{TV}^{\rho,e}$} \label{eq:PTVe}
\end{align}
We apply the \emph{soft recovery framework} developed in \cite{Flinth2017SoftTV} to prove the following result.

\begin{mainResult}[Streamlined version of Theorem \ref{th:main}] Let $\mu_0$ be as above. Fix an arbitrary angle $\theta_0 \in \set{\theta_1, \dots \theta_s}$. Assume that the $(\theta_\ell)_{\ell=1}^s$ obey a separation condition, and that for some $R>0$,
\begin{align*}
 \gamma(R) + R^{-k} \leq C \abs{c_{\theta_0}},
\end{align*}
where $\gamma(R)$ is  a certain parameter (dependent on $R$) related to the antenna design, and $C$ is a universal constant. 

Then any minimizer $\mu_*$ of the program \eqref{eq:PTV} will have a peak at a point $\theta_*$ close to $\theta_0$. The same is true for the program \eqref{eq:PTVe}, whereby the quality of the proximity guarantee will depend on the relative power $\abs{c_{\theta_0}}$, the choice of the parameter $R$, and another parameter $\beta(R)$ related to the design.
\end{mainResult} 

The parameters $\gamma(R)$ and $\beta(R)$  are related to the covering properties of the difference set $(\Delta_\ell-\Delta_k)_{k, \ell=1}^m$. We will calculate them for a few designs, both uniform linear ones as well as a circular one.

One should note that this type of soft recovery result is not as strong as the exact recovery results shown in e.g. \cite{tan2014direction}. Our result however applies to much more general antenna geometries -- and guarantees an approximate recovery of the angles of arrival. In practice, such a recovery is often sufficient.

\subsection{Notation}
Let us end this introduction by introducing some notation. ``$A \leqsim B$'' means that the entity $A$ can be upper bounded by $C \cdot B$, where $C$ is a universal constant. We denote the set of $m \times m$ Hermitian \textit{positive semi-definite} matrices with $\herm(m)$.  
It will often be convenient to identify univariate functions defined on the torus $\T = \R / \Z$ with functions defined on the sphere $\sph^1$. To be concrete, let $\angle$ denote inverse of the map $\T \to \sph^1, \omega \mapsto(\cos(\omega), \sin(\omega))$. Then we can identify $a : \T \to \C$ with $\tilde{a}: \sph^1 \to \C, \theta \mapsto a(\angle\theta)$. We will furthermore use the notation $\angle(\theta, \theta')= \angle(\theta)-\angle(\theta')$. Note that since all angles are elements of the torus, they are only defined up to an integer multiple of $2\pi$. Hence, we can, and will, always assume that they lie in the interval $[-\pi, \pi]$.

\section{Theory} \label{sec:theory}
 
In this section, we will present our main results. We will leave out many technical details, and postpone proofs to Section \ref{sec:proofs}. First, we will present the measurement model. Then we will revise the soft recovery framework from the recent paper \cite{Flinth2017SoftTV} of one of the authors, which subsequently will be used to prove the main result. Finally, we will estimate the values of the quality parameters for two antenna placement designs.

\subsection{Physical Model}

As has been outlined, we consider a set of antennas located at positions $\Delta_j \in \R^2$, $j=1, \dots m$. They collectively measure the superposition of $s$ planar waves having the same frequency $f$ and arriving from a set of distinct directions $\theta_\ell \in \sph^{1}$, $\ell = 1, \dots, s$.
 We define the response of an array element located at $\Delta$  to a planar wave coming from the direction $\theta \in \bS^1$ by $\epsilon_\Delta(\theta)=e^{i \frac{2\pi}{\lambda} \inp{\theta}{\Delta}}$, where $\lambda=\frac{c_0}{f}$ denotes the wave length with $c_0$ being the speed of the light. In this paper, we  mainly focus on narrow-band DoA estimation where the wave length $\lambda$ does not change over the whole bandwidth of the signal. Also, we normalize the array geometry by  $\frac{\lambda}{2\pi}$ where, for simplicity, we still denote the normalized array locations $\frac{2\pi \Delta_\ell}{\lambda}$ by $\Delta_\ell$, and the corresponding array response by $\epsilon_\Delta (\theta)=e^{i \inp{\theta}{\Delta}}$. Denoting by $w_\ell(t)$ the complex gain and by $\theta_\ell \in \bS^1$ the DoA of the $\ell$-th planar wave, the received signal at the $k$-th array element is
\begin{align}\label{rec_model}
r_k(t)=\sum_{l=1}^s w_\ell(t) \epsilon_{\Delta_k} (\theta_\ell)+ n_k(t),
\end{align}
where $n_k(t)$  denotes additive noise at the array element at time slot $t$. 
Also, note that in \eqref{rec_model} we made the implicit assumption that the coefficients $\{w_\ell(t)\}_{\ell=1}^s$ might vary quite fast in time $t$ but the angles $\theta_\ell$ remain stable for quite a long time. This is satisfied in almost all practical DoA estimation problems. Hence, the received signal at consecutive time slots are jointly sparse in the angle domain. In this paper, we assume that we have access to the array signal at $T$ time slots $\{r(t)\}_{t=1}^T$.
 We then form the empirical covariance matrix $\frac{1}{T} \sum_{t=1}^T r(t) r(t)^*$. We assume that $T$ is sufficiently large such that the empirical covariance matrix converges to  its expected value
\begin{align*}
	\erw{(r r^*)_{k,j}} = \sum_{\ell=1}^s \exp(i \sprod{\Delta_k-\Delta_j, \theta_\ell})c_{\theta_\ell} + \erw{nn^*}= \int_{\sph^{1}} \epsilon_{\Delta_k-\Delta_\ell}(\theta) d\mu_0(\theta)+\erw{nn^*}= (M\mu_0)_{k,j} + N,
\end{align*}
where $c_{\theta_\ell}=\bE[|w_\ell(t)|^2]$ denotes the strength of the $\ell$-th incoming wave and where we defined the measure $\mu_0 = \sum_{\ell=1}^s c_{\theta_\ell} \delta_{\theta_\ell}$, the noise matrix $N= \erw{nn^*}$, and the operator 
\begin{align*}
 	M : \calM(\sph^1) \to \herm(m), \mu \mapsto \left( \int_{\sph^1} \epsilon_{\Delta_k-\Delta_j}(\theta) d\mu(\theta)\right)_{k,j =1}^m.
\end{align*}
This is the measurement operator  $M$ which will be used in the sequel.

\begin{figure} 
\centering
\includegraphics[height=5cm]{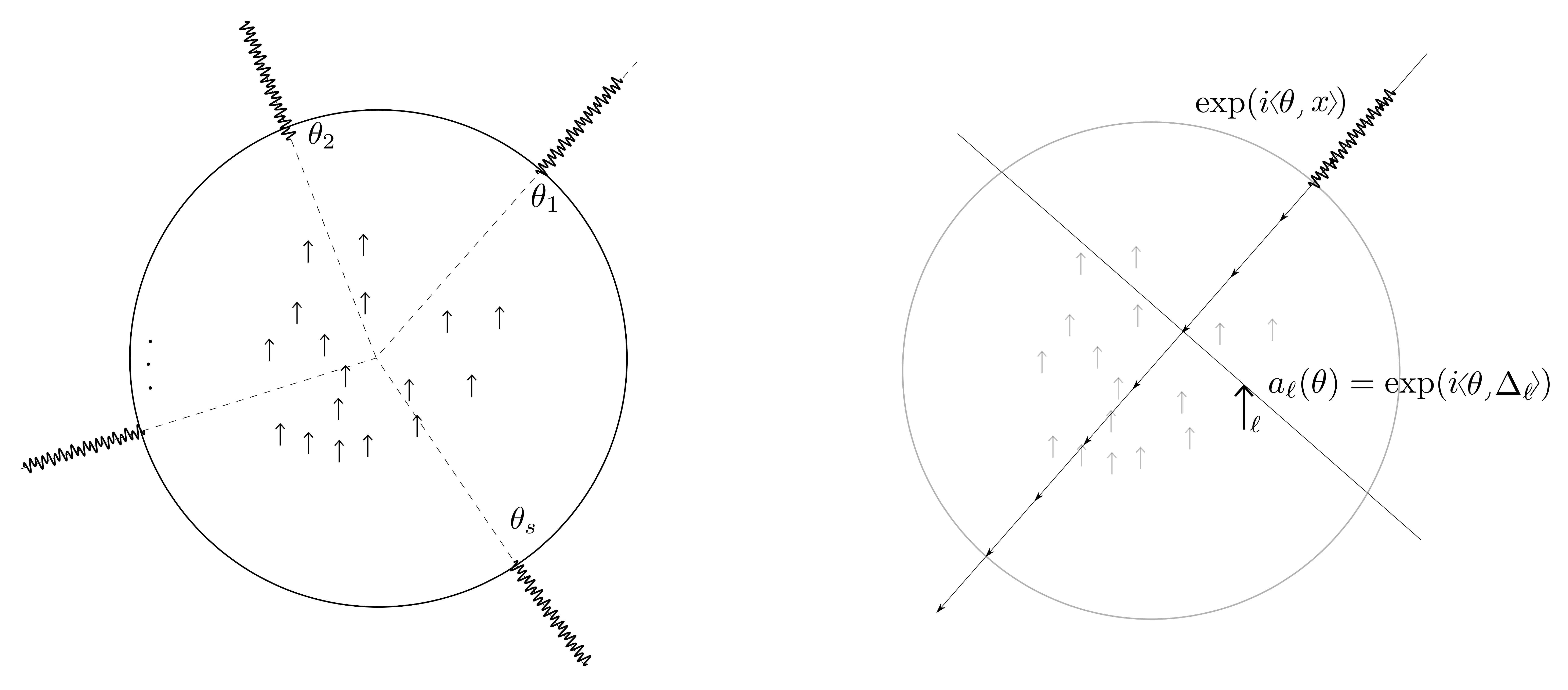}
\caption{Left: An antenna array measures a signals coming from distinct directions $\theta_j$. Right: The plain wave coming from direction $\theta$ is received as  as $\exp(i \sprod{\Delta_\ell , \theta}  )$ at antenna $\ell$.\label{fig:Model}}
\end{figure}

 \subsection{Soft Recovery}
As has been previously explained, this paper will apply the \emph{soft recovery} framework, which was developed in \cite{Flinth2017SoftTV} by one of the authors, to prove that the programs \eqref{eq:PTV} and \eqref{eq:PTVe} can approximately recover peaks in sparse atomic measures.  The framework does not only apply to $TV$-minimization programs like \eqref{eq:PTV} and \eqref{eq:PTVe}, but also more generally to the recovery of signals with \emph{sparse decompositions in dictionaries}.

In this section, we will briefly outline the definitions and main results from \cite{Flinth2017SoftTV}. We will leave out quite a few technical details -- the interested reader is referred to said publication.%
	
	\begin{defi} \cite[Definition 2.1, 2.3]{Flinth2017SoftTV}
	Let $\calH$ be a Hilbert space and $I$ a separable and locally compact metric space.
	
	\begin{enumerate}
	\item A system $(\vphi_x)_{x \in I} \sse \calH$ is called a \emph{dictionary} if the \emph{test map}
	\begin{align*}
		 T : \calH \to \calC(I), v \mapsto \left( x \mapsto \sprod{v, \vphi_x}_\calH \right),
	\end{align*}
	whereby $\calC(I)$ is the space of continuous functions on $I$, is bounded.
	\item For a dictionary, we define the \emph{dictionary operator}
	\begin{align*}
		D: \calM(I) \to \calH, \mu \mapsto D\mu = \int_I \vphi_x d\mu(x)
	\end{align*}
	through duality: $\sprod{D\mu, v} =  \int_I \sprod{\vphi_x,v} d\mu(x)$.
	
	\item The atomic norm $\norm{v}_\calA$ with respect to $(\vphi_x)_{x \in I}$ of an element $v \in \calH$ is defined as the optimal value of the optimization problem
	\begin{align} \label{eq:atomicNorm}
		\min \norm{\mu}_{TV} \st D\mu =v.
	\end{align}
	\item A signal $v \in \calH$ has a \emph{sparse decomposition} in the dictionary $(\vphi_x)_{x\in I}$ if there exists a finite set of points $(x_i)_{i=1}^s$  in $I$ and scalars $(c_i)_{i=1}^s$ in $\C$ such that
	\begin{align*}
		v = \sum_{i=1}^s c_i \vphi_{x_i} = D\big( \sum_{i=1}^s c_i \delta_{x_i}\big).
	\end{align*}
	\end{enumerate}
	\end{defi}

	It is possible to prove that the optimization problem \eqref{eq:atomicNorm} for each $v$ with finite atomic norm has a minimizer $\mu_v$. We call such a measure an \emph{atomic decomposition of $v$.} The main result of \cite{Flinth2017SoftTV} provides a dual certificate condition guaranteeing the following: Assume that a signal $v_0$ has an atomic component at a point  $x_0$. The sufficient condition stated in \cite{Flinth2017SoftTV}, which depends on both the strength and position of the peak, then guarantees that any atomic decomposition of any minimizer $v_*$ of \eqref{eq:PTV} with $b=Mv_0$ has a point $x_*$ in its support which in a certain sense is close to $x_0$. A similar statement for the regularized problem \eqref{eq:PTVe}, where $b \approx Mv_0$, was later provided in \cite{FlinthHashemi2017Thermal}. The  results can be summarized as follows.
	
	\begin{theo}[Streamlined version of Theorem 3.2, \cite{Flinth2017SoftTV} and Theorem V.2, \cite{FlinthHashemi2017Thermal}] \label{th:softOriginal}
	 Let $(\varphi_x)_{x \in I}$ be a normalized dictionary for $\calH$ and $M : \clH  \to \C^m$ be continuous. Let further $v_0 \in \calH$ 
be given through
\begin{align*}
v_0 = c^0 \varphi_{x_0} + D(\mu_c)
\end{align*} 
for a scalar $c^0_{x_0}$ and a measure $\mu_c$ such that 
\begin{align}
\norm{v_0}_\calA = \norm{c^0_{x_0} \delta_{x_0} + \mu_c}_{TV} =1 \label{eq:normalization}
\end{align}
Let $\sigma \geq 0$ and  $t \in (0, 1]$. Suppose that there exists a \emph{soft certificate}, i.e. a $\nu \in \ran M^*$ with
\begin{align}
\int_I \sprod{\vphi_x,\nu} d(c_{x_0}^0 + \mu_c) &\geq 1 \label{eq:Ankare1} \\
\abs{ \sprod{\nu, \vphi_{x_0}}} &\leq \sigma\label{eq:atPoint2} \\
\sup_{x \in I} \abs{\sprod{\nu - \sprod{\nu,\vphi_{x_0}}\vphi_{x_0}, \vphi_x}} &\leq 1-t.\label{eq:orthCompSameSub3}
\end{align}
Then for any atomic decomposition $\mu_*$ of any solution of $v$ of \eqref{eq:PTV} with $b = M\mu_0$, there exists a point $x_* \in \supp \mu_*$ with
\begin{align*}
	\abs{ \sprod{\vphi_{x_*}, \vphi_{x_0}}} \geq \frac{t}{\sigma}.
\end{align*}
Furthermore, for  any atomic decomposition $\mu_*$ of any solution of $v$ of \eqref{eq:PTVe} with $\norm{b - M\mu_0}_2 \leq \overline{e}$,
	there exists a point $x_* \in \supp \mu_*$ with
\begin{align*}
	\abs{ \sprod{\vphi_x, \vphi_{x_0}}} \geq \frac{t}{\sigma} - \frac{ 2\norm{p}_2 \overline{e} +(\rho-1)}{\rho\sigma}.
\end{align*}
where $p$ is a vector with $\nu = M^*p$.
\end{theo}

Let us make a pair of remarks to the assumptions of Theorem \ref{th:softOriginal}.

\begin{rem}

\begin{enumerate}
\item The measure $\mu_c$ can be chosen quite freely. For a signal with sparse decomposition $\sum_{i=1}^s c_i \vphi_{x_i}$, it can canonically be chosen as $\sum_{i\neq i_0} c_i \vphi_{x_i}$. It however also allows for a smaller non-sparse part, i.e. a slight model mismatch (compare with the notion of approximately sparse signals in the sparse recovery literature). 

\item The assumption \eqref{eq:normalization} is of normalization nature, and is only made for convenience. In the following we will always assume implicitly that this assumption holds. 

\item Theorem \ref{th:softOriginal} does allow for the atomic decomposition $\mu_*$ of the minimizer $v_*$ to not be a linear combination of $\delta$-peaks. Note however that there always exists at least one solution of the latter type (see Section \ref{sec:numerics} for more details).
\end{enumerate}
\end{rem}
	
	After having briefly presented the general framework, let us see how to apply it to our setting. We first need to embed the Radon measures $\calM(\sph^1)$ into a Hilbert space $\calE$. The general method for doing this in the case of $\calM(\R)$ was outlined in \cite[Section 4.4]{Flinth2017SoftTV}. We adopt basically the same idea here, namely to convolve the measures with an $L^2$-normalized, at least continuous, filter $\phi$, and then define $\calE$ as the space of all such convolutions. Note that convolution of two functions $f$ and $g$ living on $\sph^1$ is performed by identifying $\sph^1$ with the torus $\T \eqsim (-\pi, \pi]$:
\begin{align*}
	f * g ( \theta ) = \int_{\T} f(\omega) g(\angle(\theta)-\omega) d\omega.
\end{align*}

\begin{defi}
	Let $\phi \in L^2(\sph^1) \simeq L^2(\T)$ be an $L^2$-normalized filter. We define $\calE$ as the Hilbert space
	\begin{align*}
		\calE = \set{ v \in \calD'(\sph^1) \ \vert \ v * \phi \in L^2(\sph^{1})},
	\end{align*}
	with scalar product $\sprod{v,w}_\calE = \sprod{v * \phi, w*\phi}_2$. Also, $\calD'(\sph^1)$ denotes the space of distributions on $\sph^1$.
\end{defi}

This definition of $\calE$ makes the soft recovery framework applicable, as is shown in the following lemma:
\begin{lem} \label{lem:EWellDefined}
\begin{enumerate}
		\item $\calM(\sph^1)$ is embedded in $\calE$.
		
		\item The set $(\delta_\theta)_{\theta \in \sph^1}$ is a normalized dictionary in $\calE$.
		
		\item The atomic norm with respect to $(\delta_\theta)_{\theta \in \sph^1}$ is the TV-norm.
	\end{enumerate}
\end{lem}
Under certain conditions on the filter function $\phi$, the map $M: \calE \to \herm(m)$ becomes continuous. To keep the exposition brief, we choose to omit these conditions (and the subtle problems that they cause) at this point, and instead refer to Section \ref{sec:proofSoft}. The only thing we need to know at this point is that, for $M \in \N$ arbitrary, it is possible to construct a filter $\phi$ with autocorrelation function $a= \phi * \overline{\phi}$ decaying quite quickly as $\angle{\theta}$ moves away from zero, 
while still securing that $a \in \calC^s(\Omega)$ for quite large values of $s$.

It is not hard to show that $\sprod{\delta_\theta, \delta_{\theta'}}_\calE = a(\angle(\theta,\theta'))$ and $\sprod{\nu, \delta_\theta}_\calE = (\nu * a)(\theta)$. Hence, by using  exactly the  same techniques as in \cite[Sec 4.4]{Flinth2017SoftTV}, we arrive at the following:

\begin{cor} \label{cor:specialization}
Define $F_a : \calE \to \calC(\sph^1)$ through $F_a(\nu) = \nu * a$, and let $\mu_0$ have the form $c_0 \delta_{\theta_0} + \mu_c$. The existence of a $\nu$ satisfying conditions \eqref{eq:Ankare1}-\eqref{eq:orthCompSameSub3} are equivalent to the existence of a $g \in \ran F_a M^*$ with
\begin{align}
	\int_{\sph^1} \re(g(\theta)) d(c_{\theta_0} \delta_{\theta_0} + \mu^c)  &\geq 1 \label{eq:AnkareSpecial1} \\
	\abs{g(\theta_0)}&\leq \sigma \label{eq:atPointSpecial2} \\
	\sup_{x \in \sph^1} \abs{g(\theta) - g(\theta_0)a(\angle(\theta,\theta_0))} &\leq 1-t. \label{eq:orthCompSameSubSpecial3} 
\end{align}
In particular, they will imply the existence of a $\theta_* \in \sph^1$ in the support of any minimizer\eqref{eq:PTV} (for $b=M\mu_0$) with
\begin{align*}
	\abs{a(\angle(\theta_*,\theta_0))} \geq \frac{t}{\sigma},
\end{align*}
and the existence of a $\theta_* \in \sph^1$ in the support of any minimizer \eqref{eq:PTV} (for $ \norm{b-M\mu_0}_2 \leq \overline{e}$) with
\begin{align*}
		\abs{a(\angle(\theta_*,\theta_0))} \geq \frac{t}{\sigma} - \frac{ 2\norm{p}_2 \overline{e} +(\rho-1)}{\rho\sigma}.
\end{align*}
where $p$ is a vector with $g = F_aM^*p$.

\end{cor}

	\subsection{Main Result} 
	
	We are now almost ready to present the main result. Before that, we need to introduce two geometrical parameters of antenna designs, the so called \emph{quality parameters}.
	
	\begin{defi} Let $(\Delta_\ell)_{\ell=1}^m$ be a set of antenna placements. 
	\begin{enumerate}
	\item We call a collection of disjoint (up to sets of measure zero) sets $\clI:=(I_j)_{j=1}^n$ a \emph{covering associated to $(\Delta_\ell)_{\ell=1}^m$} if every set contains exactly one point $q_j$ of the difference set $(\Delta_\ell - \Delta_k)_{k, \ell=1}^m$. 
	
	\item The covering $\clI$ has the \emph{$R$-covering property} if $\cup_{j=1}^n I_j$ contains the closed ball $B_R(0)$ with radius $R$ centered in the origin. It has the \emph{centroid property} if the centroid of each $I_j$ is the corresponding point $q_j$ in the difference set.
	
	\item The $\beta(R)-$ and $\gamma(R)-$constants, or \emph{quality parameters}, of a covering $(I_j)$ are equal to $\infty$ if it does not have the $R$-covering property, and else are given through
	\begin{align*}
	 \beta(R) &= R \sqrt{\max_k \abs{I_k \cap B_R(0)}} \\
	\gamma(R) &=   \sum_{i=1}^N \diam\left( I_i \cap B_R(0) \right)^2 \abs{I_i \cap B_R(0)}.
	\end{align*}
	\end{enumerate}
	We say that a set of antenna placements has quality parameters $(\beta(R),\gamma(R))$ if there exists a covering associated to the set \emph{obeying the centroid property} with the same quality parameters.
	\end{defi}
	
	\begin{rem}
	All antenna placements have finite quality parameters for any $R>0$. To prove this, it	suffices to show that there exists a covering associated to the antenna placement having the $R$-covering and centroid property. 
	
	Towards this goal, let $r>0$ be so small  that the balls with centers belonging to the difference set not equal to $0$ and radius $r$ are pairwise disjoint. Now if we define the collection $I_{j}$ as those balls, together with an additional set $I_0$ defined as the set theoretic difference of the ball $B_R(0)$ and the union of the balls $I_j$, $j\neq 0$, then trivially, each point in the difference set not equal to zero is the centroid of `its' set. Furthermore, the centroid of the set $I_0$ is equal to zero, due to symmetry  of the difference set (if $(\Delta_k - \Delta_\ell)$ is one of the points of the difference set, $(\Delta_\ell - \Delta_k)$ will also be).
	\end{rem}
	
	As the name suggests, the quality parameters represent a measure for the ability of the programs \eqref{eq:PTV} and \eqref{eq:PTVe} to recover peaks in sparse, atomic measures. The main rule of thumb is  that the smaller the quality parameters can be made, which roughly corresponds to the associated cover being more uniform, the better. This is made precise by the following result, which is the main theorem of this paper.

\begin{theo}\label{th:main}
	Let $\mu_0$ be a measure of the form $c_{\theta_0} \delta_{\theta_0} + \mu_c$ with $\norm{\mu}_{TV}=1$. Suppose that the support of $\mu_c$ obeys
	\begin{align*}
		\sup_{\theta \in \supp \mu_c} \abs{a(\angle(\theta,\theta_0))} \leq \frac{\abs{c_{\theta_0}}}{6},
	\end{align*}
	where $a$ is the autocorrelation of the filter $\phi$ associated with $\calE$. Assume that $a \in \calC^{k}$ and that  the antenna placement design has quality parameters $(\beta(R),\gamma(R))$, whereby $\gamma(R)$ obeys
	\begin{align}
		\abs{\sph^1}(K\gamma(R) + C R^{-k}) <\frac{\abs{c_{\theta_0}}}{6}, \label{eq:measurements}
	\end{align}
	where $K$ and $C$ are universal constants.
	
	Then any minimizer $\mu_*$ of the program \eqref{eq:PTV} with $b=A\mu_0$ has a point $\theta_*$ in its support with
	\begin{align}
		\abs{a(\angle(\theta_*,\theta_0))} \geq  \frac{3\abs{c_{\theta_0}}}{8} \label{eq:supp}
		\end{align}
		
		In fact, under the same condition, any  minimizer $\mu_*$ of \ref{eq:PTVe} with $\norm{b-A\mu_0}_2\leq \overline{e}$  has a point $\theta_*$ in its support with
		\begin{align}
			\abs{a(\angle(\theta,\theta_0))} \geq  \frac{3\abs{c_{\theta_0}}}{8} - \frac{6\beta(R) \overline{e} + 2(\rho-1)\abs{c_{\theta_0}} }{3\rho},
		\end{align}
		where $\Lambda$ is a universal constant and $\beta(R)$ is the $\beta$-constant of the antenna placement design.
\end{theo}  
  
  Let us make a couple of remarks concerning this theorem.
  \begin{rem} \label{rem:toMain}
  \begin{enumerate}
\item If we choose a filter $\phi$ for $\calE$ leading to a quickly decaying autocorrelation function $a$, a statement of the form $\abs{a(\angle(\theta , \theta_0))}$ being smaller than a certain threshold will be satisfied if $\theta$ is far away from $\theta_0$.  Hence the application of the theorem naturally calls for a \emph{separation condition}. Statements such as these are well-known in the sparse measure recovery literature (see for instance \cite{candes2014towards, DuvalPeyre2015}). Correspondingly, $\abs{a(\angle(\theta,\theta_0))}$ being larger than a certain threshold implies that $\theta$ and $\theta_0$ are close. Hence, the statement provided by the theorem corresponds to a proximity guarantee of the recovered peak $\theta_*$ to the ground truth peak $\theta_0$.
  
 \item  For most designs, the number of antennas $m$ is a parameter which is subject to change. This will typically lead to new quality parameters -- ideally, they should decay fast with $m$. Let us assume that the parameter $\gamma(R)$ can be estimated above by $\Gamma R^\alpha m^{-\beta}$, where $\Gamma$, $\alpha$ and $\beta$ are positive constants. \emph{If $R$ can be chosen freely}, the fact  $\min_R C R^\alpha m^{-\beta} + R^{-k} \sim m^{-\tfrac{\beta k}{\alpha +k}}$ reveals that a sufficient condition for the bound \eqref{eq:measurements} to hold is
  \begin{align*}
   m \geqsim \abs{c_{\theta_0}}^{\frac{\alpha+k}{k\beta}}.
  \end{align*}
  
	The bound \eqref{eq:measurements} can hence be seen as an indirect bound on the number of antennas needed to secure approximate recovery of a signal component with relative power $c_{\theta_0}$.
	
	Since $k$ can in theory be chosen as large as we want, we can streamline even a bit more: If $R$ can be adjusted along with $m$, a bound $\gamma(R) \leqsim m^{-\beta}$ \emph{ roughly} corresponds to an asymptotic of $m \geqsim \abs{c_{\theta_0}}^{\frac{1}{\beta}}$.
	
	There are also situations where $R$ cannot be adjusted along with $m$ -- as we will see in the sequel, an increasing $R$ often corresponds to the antennas being spread over a larger area, which may not be feasible. In this case, the term $CR^{-k}$ in \eqref{eq:measurements} indicates that there exists a fundamental resolution limit, which is independent of the number of antennas.
	\end{enumerate}
	\end{rem}

  \subsubsection{The Quality Parameters $\beta(R)$, $\gamma(R)$} 
  
  In this section, we will provide bounds on the values of the quality parameters for some concrete designs. 
  
  \paragraph{One-dimensional Uniform Linear Arrays.}
The arguably simplest antenna placement design is that of a \emph{uniform linear array}:
\begin{align*}
	\Delta_k= \left(0,\frac{k R}{m}\right) \in \R^2, \quad k=0, \dots, m-1.
\end{align*}
The corresponding difference set is given by $\left((0,\tfrac{Rk}{m})\right)_{k=-(m-1)}^{(m-1)}$. As is made clear by Figure \ref{fig:CoPrimeCovering1}, an associated covering $(I_j)_{j=-(m-1)}^{m-1}$ with the centroid property must satisfy $\diam{I_j \cap B_R(0)} \sim R$ for most $j$. Since at the same time $\sum_j\abs{I_j \cap B_R(0))} = \abs{B_R(0)}$, we must also have $\abs{I_j \cap B_R(0)} \sim R^2m^{-1}$. This leads to quality parameters
\begin{align*}
	 \beta(R) &\sim R^2 \\
	\gamma(R) &\sim   \sum_{i=(m-1)}^{(m-1)} R^2 \cdot R^2 m^{-1} \asymp R^4
\end{align*}
 Hence, the quality parameters do not decay with a growing number of antennas $m$. This reflects the fact that uniform linear arrays are only able to resolve AoA's in a limited angular region. Let us explain this in a bit more detail: The standard way to treat uniform linear arrays is to transform the recovery problem into a Fourier inversion via the transformation $\sin(\theta) = t$, and then recover the parameters $t_i$ rather than the angles $\theta_i$. The powerful theory of the Fourier case leads to excellent performance and stability guarantees. However, due the nature of the $\arcsin$-function, a small error in the estimation of a position $t_i$ when $\abs{t_i}\approx 1$ lead to a large error for the estimation of the corresponding angle $\theta_i = \arcsin(t_i)$. Since Theorem \ref{th:main} makes claims about the latter type of recovery, this subtlety explains the relatively bad asymptotic behavior of the quality parameters.

 \begin{figure}
 \centering
	\includegraphics[width=5cm]{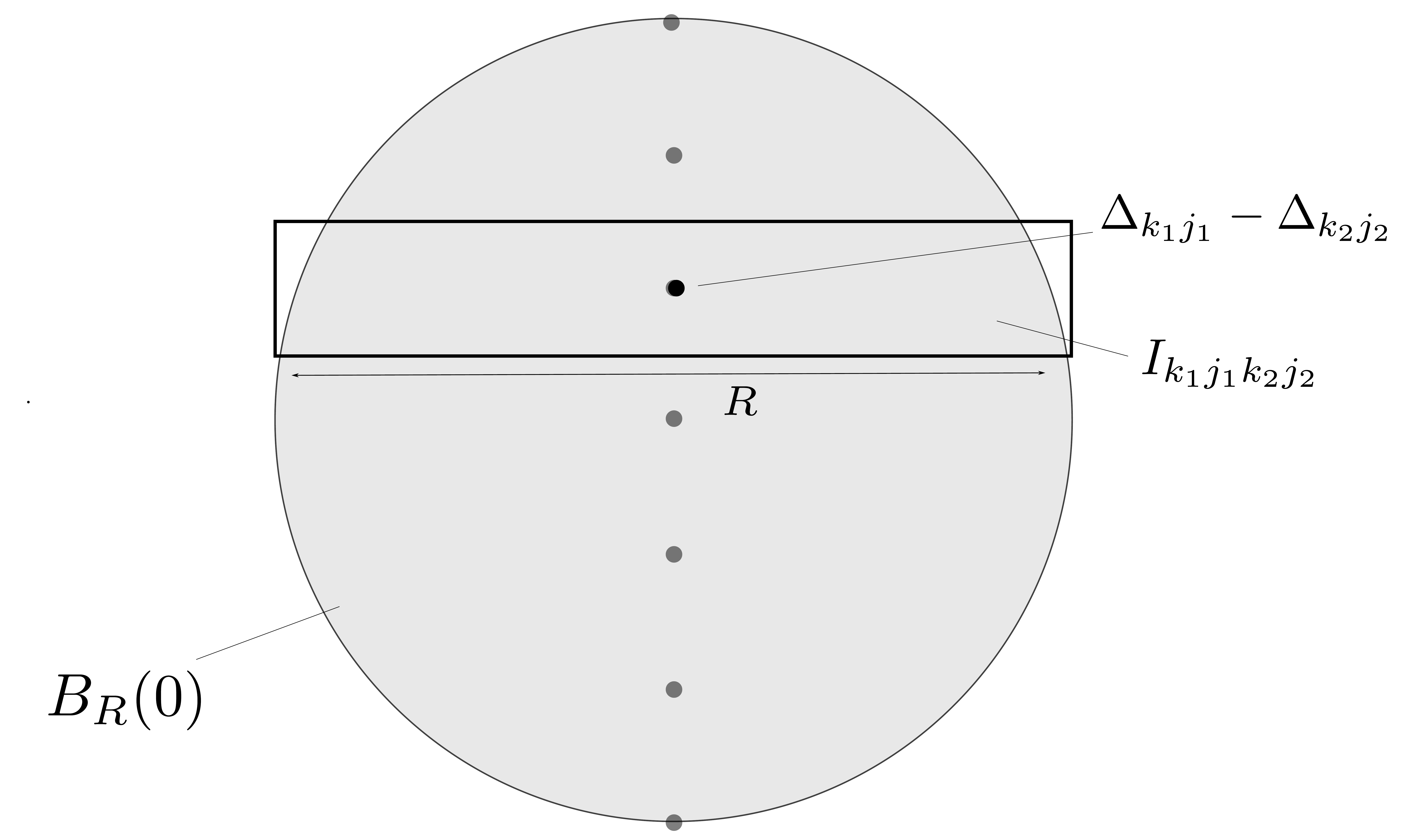}
	\caption{The covering of the set induced by a one-dimensional uniform linear array. \label{fig:CoPrimeCovering1}}
 \end{figure}

 \paragraph{Two-Dimensional Linear Arrays}
 
 The last paragraph showed that traditional one-dimensional design have bad quality parameters. This naturally poses the question if two dimensional designs are better.
 
First, let us consider a simple, two-dimensional uniform linear array:
		\begin{align} \label{eq:2DLinearArray}
			\Delta_{kj}= \frac{R}{P}\begin{bmatrix}
			k \\ j
\end{bmatrix}, 	\quad k,j \in \set{0, \dots, P-1}		 
		\end{align}
The corresponding difference set is given by $\frac{R}{P} \set{-P-1, \dots, P-1}$. An associated covering is given by:
		\begin{align*}
			I_{q_1, q_2} = \frac{R}{P} \left( \begin{bmatrix}
				q_1 \\ q_2
			\end{bmatrix} +  \left[-\frac{1}{2},\frac{1}{2}\right]^2 \right), \quad q_i \in \set{-(P-1), \dots, (P-1)}
		\end{align*}
		This covering trivially has the $R$-covering property, as well as the centroid property. Furthermore, it is easily seen that
		\begin{align*}
			\diam(I_{q_1, q_2} \cap B_R(0)) \leq \sqrt{2} RP^{-1}, \ \abs{I_{q_1,q_2} \cap B_R(0)} \leq R^2P^{-2},
		\end{align*}
		so that
		\begin{align*}
			\beta(R) &\leqsim R^2m^{-2} \\
			\gamma(R) &= 2 \sum_{q_1, q_2 = -(P-1)}^{P-1} R^2 P^{-2} R^2 P^{-2} = 2(M-1)^2R^4P^{-4} \leqsim R^4 P^{-2}.
		\end{align*}
	Considering the fact that the number of antennas $m$ for the design \eqref{eq:2DLinearArray}	equals $P^2$, we have proven the following proposition:
	\begin{prop} \label{prop:LinearArray}
		The uniform linear array \eqref{eq:2DLinearArray} with $m$ antennas has quality parameters obeying
		\begin{align*}
			\beta(R) \leqsim R^2m^{-1} \quad \gamma(R) \leqsim R^4 m^{-1}
		\end{align*}
	\end{prop}

\begin{figure}
	\centering
	\includegraphics[width=5cm]{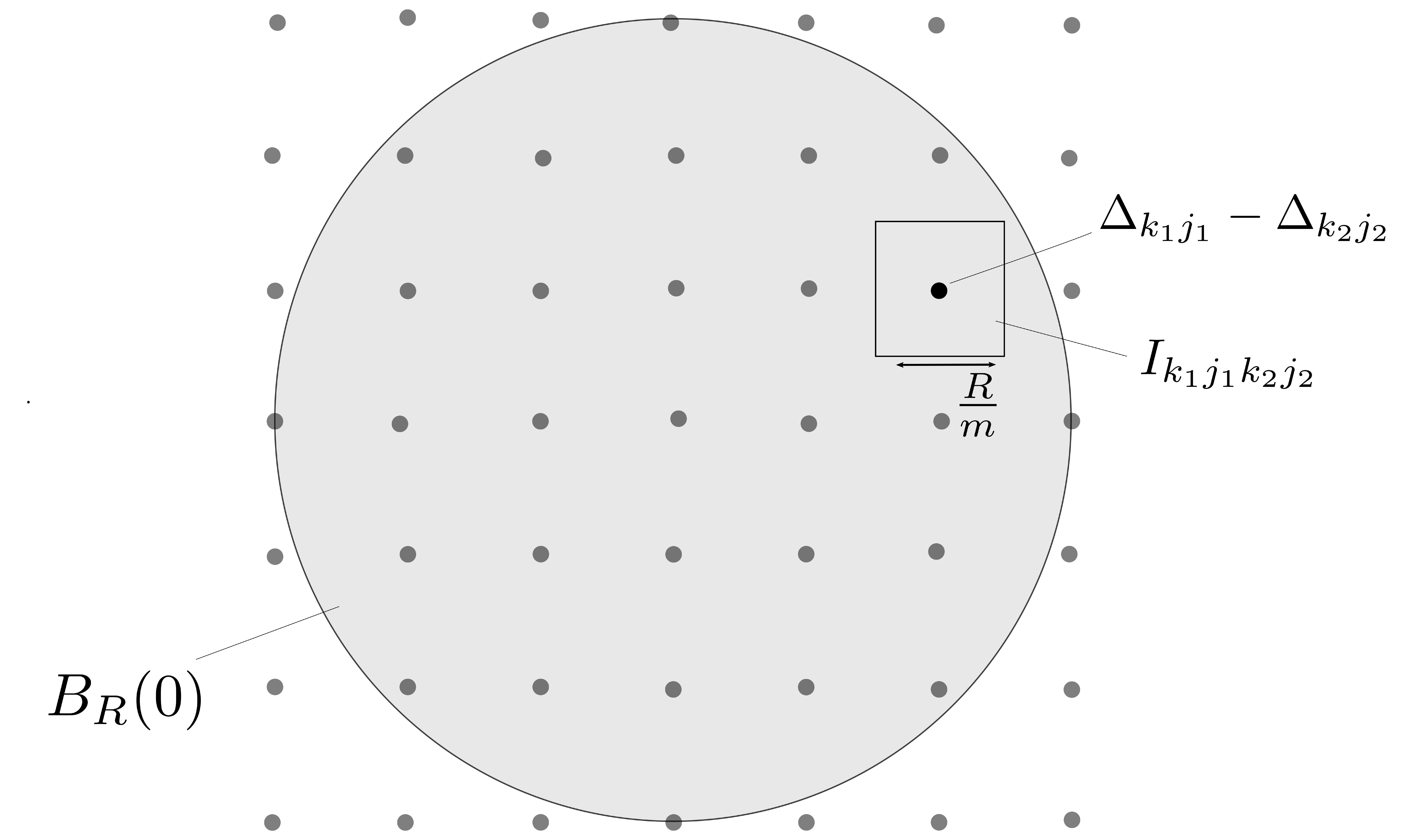}
	\caption{The covering of the set induced by a two-dimensional uniform linear array. \label{fig:CoPrimeCovering}}
\end{figure}

		A simple way to reduce the number of antennas is to use so-called co-prime sampling \cite{vaidyanathan2011sparse, pal2010nested, haghighatshoar2017massive}, i.e. to only consider values of $k$ in $j$ in \eqref{eq:2DLinearArray} in a certain set  $\calD$ with the property $\calD - \calD = \set{-(P-1), \dots, (P-1)}$. Such an antenna design, by definition, leads to the same difference set and thus also the same coverings and quality parameters. Since there exists, for certain values of $P$, difference sets with $\abs{\calD} \sim \sqrt{P}$ exist, we obtain the following simple corollary:
		\begin{cor}
			There exists subsampling patterns of two-dimensional uniform array with $m$ antennas having quality parameters obeying
			\begin{align*}
				\beta(R) \leqsim R^2m^{-2} \quad \gamma(R) \leqsim R^4 m^{-2}
			\end{align*}
		\end{cor}
		
		Note that although the subsampled arrays are more efficient with respect to the number of antennas, they are less robust to antenna failures. If one of the antennas in the full two-dimensional design is removed, the difference set does not change much (since many pairs of parameter values $(k,j)$ give rise to the same difference), leading to similar quality parameters. A corresponding loss for a co-prime subsampled linear array affects the difference set much more.
 
 \paragraph{Circular Arrays}
	Let us end this section by making a case study on \emph{circular designs}, showing that the framework we have developed also allows us to consider geometries radically different from the linear ones discussed above. 
	
	By a ``circular design'', we mean an antenna arrangement of the form:
	\begin{align}
		\Delta_k = D e^{\frac{2\pi i k}{m}}, \quad k=1, \dots, m, \label{eq:circular}
		\end{align}			
where we identified $\R^2$ with $\C$. For this set, the construction of the sets $I_k$, as well as the calculations of the $\beta(R)$ and $\gamma(R)$-designs, are a lot trickier than the ones for the co-prime design above. Therefore, we postpone them to Section \ref{sec:CircDesign}. The results are concluded in the following proposition. Notably, we arrive at asymptotics similar to the ones of the cleverly subsampled  linear arrays  from above (which only exist for certain values of $M$).

\begin{prop} \label{prop:CircularQualityParameters}
	Let $m \geq 5$. There exists a universal constant $\Theta$ such that the quality parameters of the circular design \eqref{eq:circular} with $m$ antennas fulfill \begin{align*}
		\beta(\Theta D) \leqsim D^2m^{-1}, \quad \gamma(\Theta D) \leqsim D^4m^{-2}.
	\end{align*}
\end{prop}

\section{Numerics} \label{sec:numerics}
The problems \eqref{eq:PTV} and \eqref{eq:PTVe} are of infinite-dimensional nature, and hence, solving them numerically is far from trivial. 
 In large generality, the support of a solution to a problem of the form
	\begin{align}
		\min_{\mu} \norm{\mu}_{TV} + f_b(M\mu) \tag{$\calP$} \label{eq:primal}
	\end{align}
can be found by solving the \emph{dual problem }
\begin{align}
	\max_{p} - f_b^*(p) \st \norm{M^*p}_\infty \leq 1. \tag{$\calD$} \label{eq:dual}
\end{align}
$f_b^*$ denotes the \emph{Fenchel conjugate} (see e.g. \cite{Rockafellar1970}) of the convex function $f_b$. Now, if $(\mu_*,p_*)$ is a primal-dual pair, $\supp \mu_* \sse \set{ x : \abs{(M^*p)(x)}=1}$. Hence, if the set $\set{ x : \abs{(M^*p)(x)}=1}$ consists of finitely many isolated points, they can be used for an ansatz to solve the primal problem \eqref{eq:primal} through a finite-dimensinal optimization procedure.

 Choosing $f_b$ equal to the indicator of the set $\set{b}$, we see that  \eqref{eq:PTV} is of this form, with $f_b^*(p) = -\re(\sprod{p,b})$, that is, the dual problem is given
\begin{align}
	\max \re( \sprod{p,b} ) \st \norm{M^*p}_\infty \leq 1 \tag{$\calD_{TV}$} \label{eq:DTV}
\end{align} 
As for \eqref{eq:PTVe}, remember that for each $\rho$, there exists a $\Lambda$ such that the solution of \eqref{eq:PTVe} is equal to the solution of
\begin{align*}
	\min \frac{1}{2} \normm{M\mu - b}_2^2 + \Lambda \norm{\mu}_{TV} \tag{$\calP_{TV}^{\mathrm{reg},\Lambda}$}, \label{eq:reg}
\end{align*}
 which certainly is of the form \ref{eq:primal}. The corresponding dual problem is
 \begin{align}
 	\max_p \re(\sprod{p,b}) - \frac{\Lambda}{2}\norm{p}_2^2 \st \norm{M^*p}_\infty \leq 1 \tag{$\calD_{TV}^{\mathrm{reg},\Lambda}$} \label{eq:regDual}
 \end{align}
 
The dual problems  \eqref{eq:dual} and \eqref{eq:regDual} are not easier to solve than the primal problem per se, since the constraint is still infinite dimensional. For some special examples of measurement operators $M$, these constraints can however be converted to equivalent finite-dimensional
ones. The most prominent example is when $M$ models the sampling of the Fourier transform of $\mu$ (see for instance \cite{candes2014towards,tang2013compressed}). The idea of this approach is that the constraint of the \emph{dual problem} can be rewritten as a \textit{semi-definite program} (SDP) with a linear constraint related to the dual variable $p$. Hence, the problem can be solved with off-the-shelf SDP solvers. As was shown in \cite{flinthweiss2017exact}, this strategy is in fact applicable with measurement functions equal to any trigonometric polynomials (not just the canonical ones).

For general measurement functions, such as the ones considered in this publication, there is not much hope to being able to rewrite the problems in a way similar to above, and one has to resort to discretization procedures. The most canonical way of doing this is to make an \textit{ansatz}, assuming that the support points of the solution are located on a fixed grid $D_N$:
\begin{align*}
	\mu = \sum_{\theta \in D_N} c_\theta \delta_{\theta},
\end{align*}
and solve the discretized problems
\begin{align}
	\min_c \norm{c}_1 & \st M_{D_N}c = b \tag{$\calP_{TV,D_N}$} \label{eq:PTVdisc} \\
	\min_c \normm{M_{D_N} c - b}_2 & \st \norm{c}_1 \leq \rho \tag{$\calP_{TV,D_N}^{\rho,e}$}, \label{eq:PTVedisc}
\end{align}
where we defined the maps $M_{D_N} \in \C^{\abs{D_N}} \to \herm(m)$ through
\begin{align*}
	M_{D_n} c = M\left(\sum_{\theta \in D_N} c_\theta \delta_{\theta}\right)
\end{align*}
This strategy has a problem. First, although there will always exist at least one solution of the problem which is a short linear combination of Diracs \cite[Theorem 1]{flinthweiss2017exact}\footnote{Formally, \eqref{eq:PTVe} is not of the same form as the ones treated in the cited source. It is however well-known that the problem can be rewritten into one which is of the form treated in the cited source -- more details will be given below.  }, they are not guaranteed to lay on the grid. This leads to so-called basis-mismatch problems \cite{chi2011sensitivity}, meaning that measures which are per se sparse do not have sparse representations in the discretized basis. This is in fact one of the main motivations of going over to the infinite dimensional problems \cite{tang2013compressed}.

In this section, we will outline a semi-heuristic approach to solving problems of the form \eqref{eq:PTV} and \eqref{eq:PTVe} which should give slightly better results than the naive strategy of simply solving the problems \eqref{eq:PTVdisc} and \eqref{eq:PTVedisc}. This approach relies on approximating the measurement functions with trigonometric polynomials.

\subsection{Trigonometric Approximations of the Measurement Functions }

As has already been mentioned, one situation where it is possible to transform the dual problems \eqref{eq:regDual} and or \eqref{eq:DTV} to 
finite-dimensional SDPs is the one of the measurement functions being trigonometrical polynomials, i.e. of the form
\begin{align*}
	\rho_i(\omega) = \sum_{\abs{k} \leq L} \Gamma_{k,i} e^{k\omega}, \quad \omega \in \T
\end{align*} 
In our setting, we can identify each $\epsilon_{\Delta_k - \Delta_\ell}$ with a function $$\epsilon_{k\ell}(\omega) = \exp\left( i \left(\cos(\omega)(\Delta_k-\Delta_\ell)_1 + \sin(\omega)(\Delta_k - \Delta_\ell)_2\right)\right).$$
This is certainly no trigonometric polynomial but is a  smooth  function on $\T$, so it can efficiently be approximated using trigonometric functions:
\begin{align}
	\epsilon_{k\ell} \approx \rho_{k\ell} =  \sum_{\abs{i} \leq L} \Gamma_{i,k\ell} e^{i\omega}. \label{eq:epsilonApprox}
\end{align}
The exact method of approximation is thereby not crucial. Since $\epsilon_{k\ell}$ should be regarded as a member of the dual of $\calM(\T)$, i.e $\calC(\T)$, and be approximating accordingly, the use of C{e}s\`{a}ro means of Fourier series is appropriate.

Let $\widetilde{M}$ denote the measurement operator defined by the functions $\rho_{k\ell}$. If $L$ is large enough, $\widetilde{M}$ will be close to $M$, $\norm{b-\widetilde{M}\mu_0}_2$ will not be much larger than $\norm{b-M\mu_0}_2$. Hence, solving the problem
\begin{align}
	\min \frac{1}{2} \norm{\widetilde{M}\mu - b}_2^2 + \Lambda \norm{\mu}_{TV}  \tag{$\calP_{TV, \mathrm{trig}}^{\mathrm{reg},\Lambda}$} \label{eq:Ptrig}
\end{align}
should produce a good approximation of the solution to \eqref{eq:reg}. Fortunately, the problem \eqref{eq:Ptrig}  can be exactly solved. To be concrete, the following proposition holds:

\begin{prop}
(Adapted version of \cite[Lemma 3]{flinthweiss2017exact}) The dual problem of \eqref{eq:Ptrig} can be rewritten as follows
\begin{align}
	\min_p \re(\sprod{p,b}) - \frac{\Lambda}{2}\norm{p}_2^2 \st \exists  Q\in \C^{(2L+1)\times (2L+1)},\begin{bmatrix}
                                                             Q & \Gamma p \\
                                                             (\Gamma p)^* & 1
                                                            \end{bmatrix}\succeq 0, \tag{$\calP_{\mathrm{trig SDP}}$} \label{eq:trigSPD}\\
\sum_{i=1}^{2L+2-j}Q_{i,i+j}=
\begin{cases}
1, & j=0, \\
0, & 1\leq j \leq 2L+1.                                                                                                \end{cases}\Bigg\} \nonumber.
\end{align}
\end{prop}

The scheme outlined here is summarized in Algorithm \ref{alg:TrigAppr}.

\begin{algorithm}
		\caption{Approximation by Trigonometric Polynomials} \label{alg:TrigAppr}
		\KwData{ An antenna design  $(\Delta_k - \Delta_\ell)_{k,\ell}$,(noisy) measurements $b \in \herm(m)$ and a $\Lambda>0$}
		\KwResult{ An estimate $\mu_*$ of a sparse solution to \eqref{eq:reg}.}
		
		\nl Calculate the matrix $\Gamma$ defining the trigonometric approximations \eqref{eq:epsilonApprox}.
		
		  \nl Find the solution $p_*$ of the problem \eqref{eq:trigSPD}. \
		   
		  \nl Find the zeros $\widehat{D}$ of the trigonometric polynomial $\abs{\widetilde{M}^*p_*}^2-1$. \
		  
		  \nl Find the solution $c^*$ discretized to $\widehat{D}$
		  \begin{align}
		  		\min \frac{1}{2} \norm{M_{\widehat{D}}c-b}_2^2 + \Lambda \norm{c}_1 \tag{$\calP_{TV,\mathrm{disc}}^{\mathrm{reg},\Lambda}$}
		  \end{align}
		  (again by any method for finite-dimensional convex programming). \
		  
		  \nl Output
		  \begin{align*}
		  	\mu_*= \sum_{\theta \in \widehat{D}} c_\theta^* \delta_{\theta}
		  \end{align*}
\end{algorithm}

 \begin{figure}[t]
	\centering
	\includegraphics[width=0.5\textwidth]{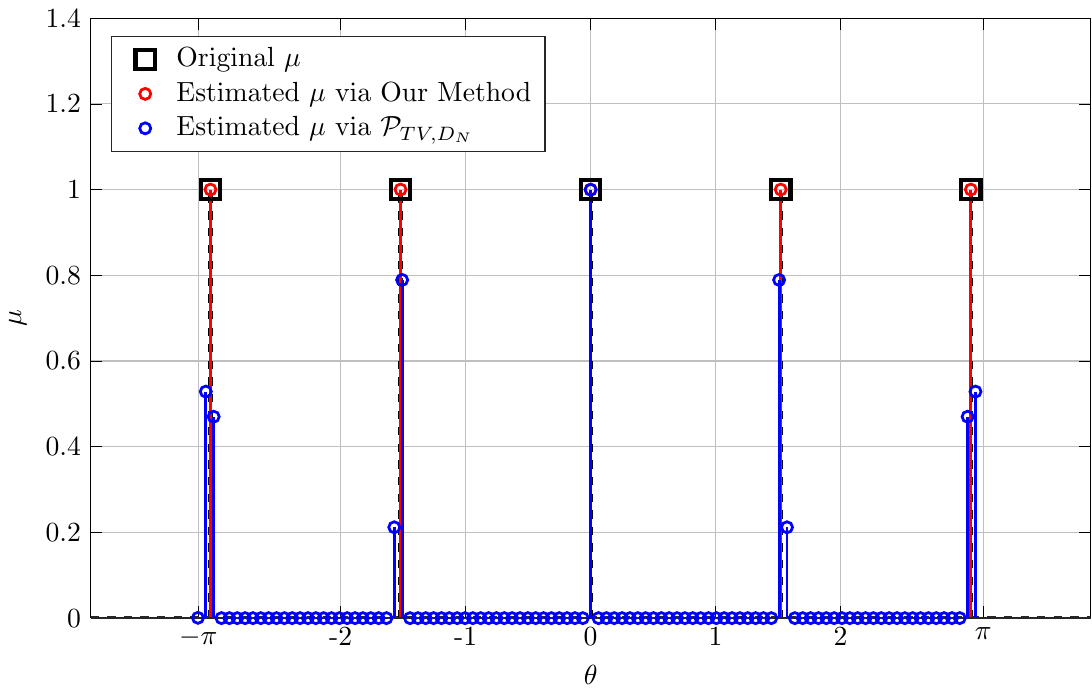}
	\caption{Estimation accuracy comparison.}
	\label{fig:mu_est}
\end{figure}
 \subsection{Experiments}
 Here we compare the performance of the technique presented in Method \ref{alg:TrigAppr} with that of the discretized problems \eqref{eq:PTV} and \eqref{eq:PTVe}. We consider a circular array with $m=17$ antennas and a radius $D =1$. The order of the trigonometric expansion in \eqref{eq:epsilonApprox} is set to $L=20$, which is a relatively small value and does not add a large computational burden. The number of discrete angular grid points in \eqref{eq:PTV} and \eqref{eq:PTVe} is set to $N=100$. The regularizing parameter $\Lambda$ in all of the methods is chosen by cross validation. 

For the first experiment, we generate a positive measure $\mu$ consisting of $5$ equispaced spikes, each with an amplitude of 1. Then we calculate the measurements $b$ by applying the operator $M$ to the generated measure. Finally, given the observation, we estimate $\mu$ using both our method and the discretized method \ref{eq:PTV}. To find the roots of the trigonometric polynomial $\abs{\widetilde{M}^*p_*}^2-1$ in step 3 of Method \ref{alg:TrigAppr}, we sample it on a uniform grid with a step size 0.001 and find those samples for which $\abs{\widetilde{M}^*p_*}^2-1 \approx 0$. Because of the limitation in machine precision, the number of such samples might be much more than the genuine roots in the solution. Therefore, we perform a local averaging on candidate samples to specify the location of a root. Figure \ref{fig:mu_est} demonstrates the results. Note that here we are in the noiseless case. As we can see, our method is able to recover the measure with high accuracy, while the discretized methods returns a measure with slightly misplaced peaks, showing a mismatch between the grid dictionary and the true sparsifying dictionary. Also because of the leakage of energy to other atoms of the discretized dictionary, the amplitude of the spikes is not well estimated.

We also study the performance of our method under various noise conditions. To this purpose, we generate a measure consisting of a spike with amplitude 1 and with random position in the interval $[-\pi,\pi]$. Then we add white Gaussian noise with different variance values to the observation matrix $b$. Given the noisy observation, we perform Method \ref{alg:TrigAppr} and obtain a estimate of the spike position. Figure \ref{fig:angular_err} shows the results in terms of absolute angular deviation, measured in degrees. As we can see even for small \textit{Signal to Noise Ratio} (SNR) values, the error can be as low as $\approx$ 1.4 degrees and for high SNR values it reduces to about 0.2 degrees. 

\begin{figure}[t]
	\centering
	\includegraphics[width=0.5\textwidth]{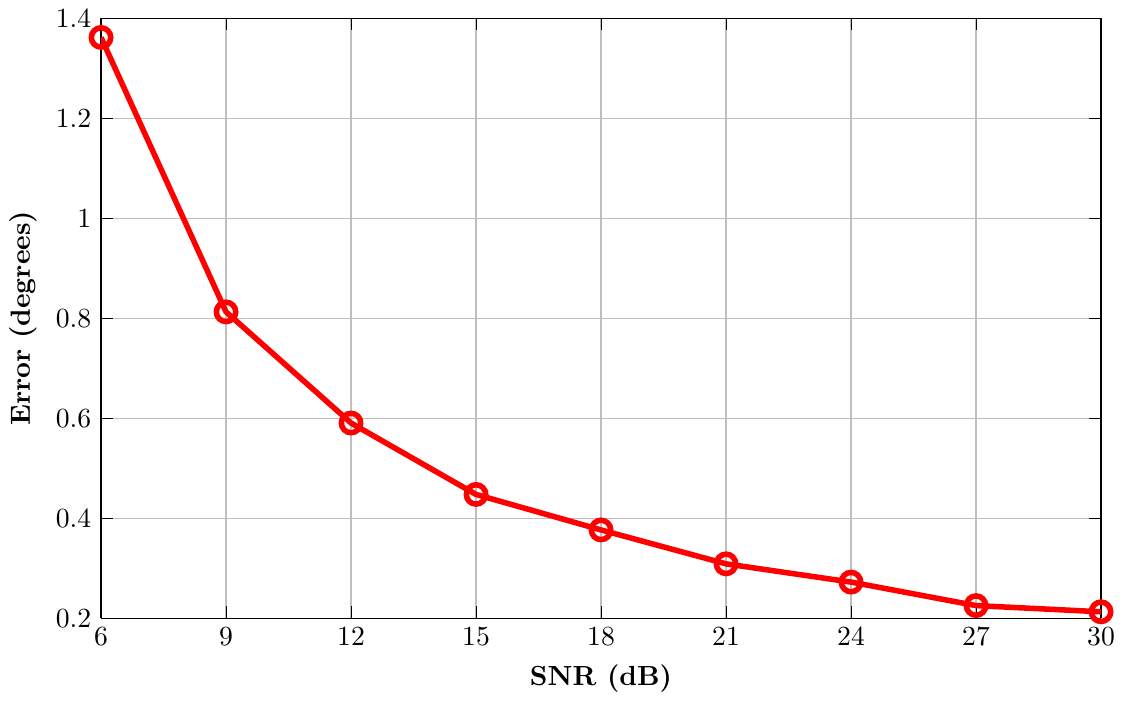}
	\caption{Absolute angular deviation of the estimate vs SNR for a measure with a single spike.}
	\label{fig:angular_err}
\end{figure}
\begin{figure}[t]
	\centering
	\includegraphics[width=.32\textwidth]{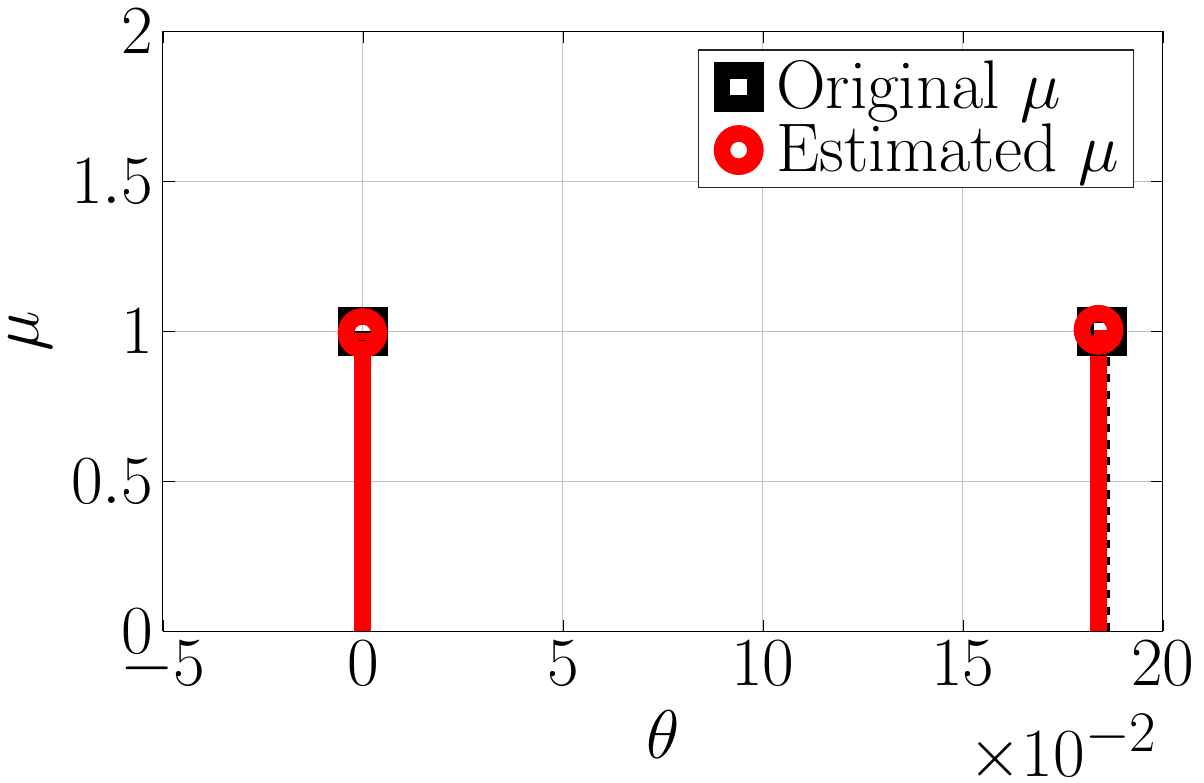}\hfill
	\includegraphics[width=.32\textwidth]{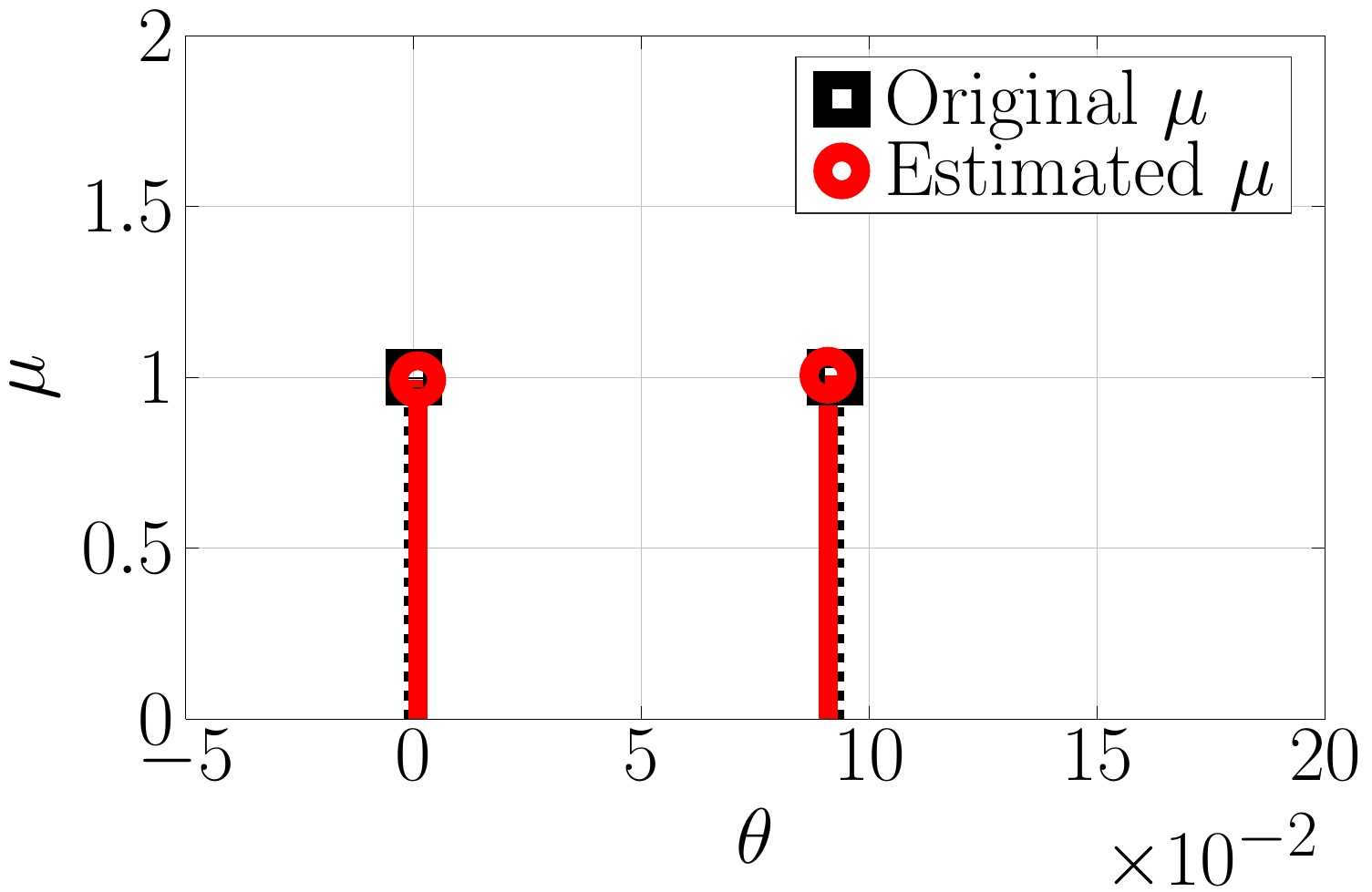}\hfill
	\includegraphics[width=.32\textwidth]{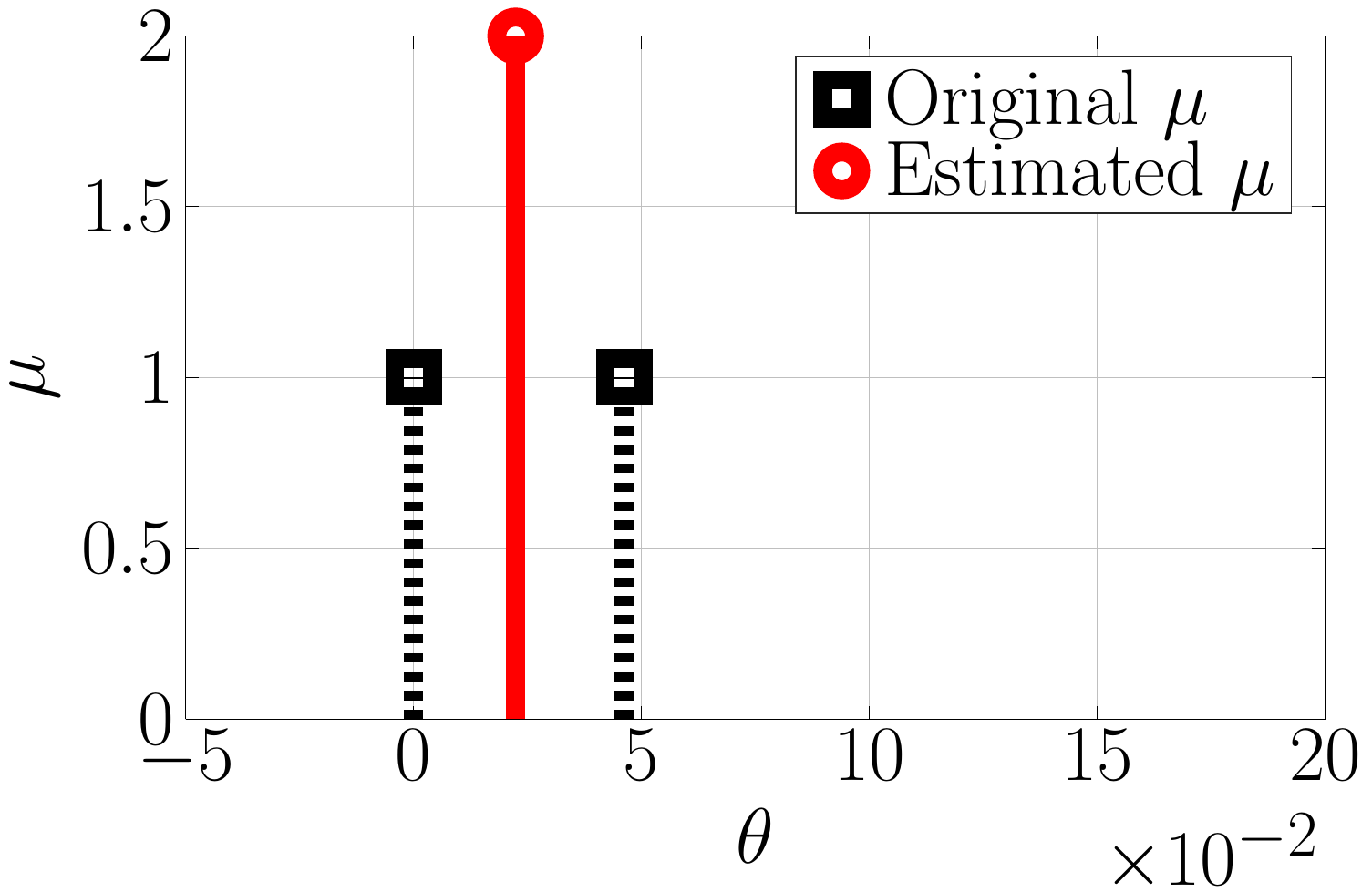}
	\caption{Estimation performance for different values of separation between two spikes in the measure $\mu$: (Left) $\theta_{\text{sep}}=\frac{\pi}{m}$, (Middle) $\theta_{\text{sep}}=\frac{\pi}{2m}$, (Right) $\theta_{\text{sep}}=\frac{\pi}{4m}$.}
	\label{fig:res_nless}
\end{figure}

An important feature to be examined in an DoA estimation method is its resolvability performance, i.e. what is the minimum separation between two spikes in the measure $\mu$ such that the method is still able to resolve them. We test this by generating a measure with two spikes, each with an amplitude equal to 1, where one of them is located at $\theta=0$ and the other one is located at $\theta = \theta_{\text{sep}}$, where $\theta_{\text{sep}}$ is the separation angle. We perform Method \ref{alg:TrigAppr} on the covariance matrix generated by this measure for different values of $\theta_{\text{sep}}$ and calculate the estimates. Figure \ref{fig:res_nless} illustrates the results in the noiseless case for three different values of $\theta_{\text{sep}}$. As we can see, the method can resolve two adjacent spikes with a separation as small as $\theta_{\text{sep}}=\frac{\pi}{2m}$. Here, this corresponds to resolving two spikes that are about 5 degrees apart, using $m=17$ antennas. For $\theta_{\text{sep}}=\frac{\pi}{4m}$ the method merges the two spikes, resulting in one spike in the middle with an amplitude of 2.

\subsection*{Acknowledgements}
This work has been funded by Deutsche
Forschungsgemeinschaft (DFG) Grant KU 1446/18~-~1.

Axel Flinth acknowledges support from the Berlin Mathematical school. He furthermore wants to thank Felix Voigtländer for interesting discussions, in particular relating to Section 4.1.3.

\bibliographystyle{abbrv}
\bibliography{bibliographyCSandFriendsMASTER}

\section{Proofs} \label{sec:proofs}
\subsection{Soft Recovery} \label{sec:proofSoft}

Before getting in to the construction of dual certificates, let us begin by checking that the  soft recovery framework is applicable at all. 

\begin{proof}[ Proof of Lemma \ref{lem:EWellDefined}]
 $1.$ For  $\mu \in \calM(\sph^{1})$ arbitrary, we have
	\begin{align*}
		\norm{\mu * \phi}_2^2 & = \int_{\sph^1} \int_{\sph^1} \int_{\sph^{1}} \phi(\angle(\theta, \omega)) \overline{\phi(\angle(\theta, \omega'))} d\mu(\omega') d\mu(\omega) d\theta \leq \int_{\sph^1} \int_{\sph^1} \abss{\int_{\sph^{1}} \phi(\angle(\theta ,\omega)) \overline{\phi(\angle(\theta ,\omega'))} d\theta} d\abs{\mu}(\omega') d\abs{\mu}(\omega)   \\
		&\leq \int_{\sph^1} \int_{\sph^1}\norm{\phi( \angle(\cdot , \omega))}_2 \norm{\phi(\angle(\cdot , \omega'))}_2  d\abs{\mu}(\omega') d\abs{\mu}(\omega)   =  \norm{\phi}_2^2 \norm{\mu}_{TV}^2
	\end{align*}
	We applied Fubini's theorem at one point. This is justified by $$\phi(\angle(\theta, \omega))\phi(\angle(\theta,\omega')) \in L^\infty(\sph^1 \times \sph^1 \times \sph^1, \lambda \otimes \mu \otimes \mu ) \hookrightarrow L^1(\sph^1 \times \sph^1 \times \sph^1, \lambda \otimes \mu \otimes \mu).$$
	
	$2.$ First, we have $\delta_\theta * \phi = \phi(\angle(\cdot, \theta))$, so the normalization follows from the normalization of $\phi$: $\norm{\delta_\theta * \phi}_2 = \norm{\phi_2}$. In order to check that $(\delta_\theta)_{\theta \in \sph^1}$ is a dictionary, we show that its test map $T$ is continuous. We have
	\begin{align*}
		\sprod{v , \delta_\theta}_\calE = \int (v*\phi)(\omega)\overline{\phi(\omega-\theta)} d\omega.
	\end{align*}
	Due to the continuity of $\phi$, $(v*\phi)(\omega)\overline{\phi(\angle(\omega,\theta'))} \to (v*\phi)(\omega)\overline{\phi(\angle(\omega,\theta))}$ for every $\omega$ when $\theta' \to \theta$. This furthermore happens under the integrable majorant $\norm{\phi}_\infty \abs{v * \phi}$ (notice that $v*\phi \in L^2(\sph^1) \emb L^1(\sph^1)$, so $\sprod{v, \delta_\theta}_\calE$ is continuous, which means that $T$ is well defined. Furthermore, $\abs{\sprod{v, \delta_\theta}_\calE} \leq \norm{v}_\calE$ for all $\theta$, so the test map is also continuous.
	
	$3.$ It suffices to prove that if $\mu = \int_{\sph^1} \delta_\theta d\nu(\theta)$, then $\mu = \nu$. For this, let $w \in \calE$ be arbitrary. We then have
	\begin{align*}
		\sprod{w, \mu}_\calE &= \sprod{w, \int_{\sph^1} \delta_\theta d\nu(\theta)}_\calE = \int_{\sph^1} \sprod{w, \delta_\theta}_\calE d \nu(\theta) = \int_{\sph^1} \int_{\sph^1} (w * \phi)(\omega) \overline{\phi(\angle(\omega,\theta))} d\omega d \nu(\theta) \\
		&= \int_{\sph^1} \int_{\sph^1} (w * \phi)(\omega) \overline{\phi(\omega- \theta)}  d \nu(\theta) d\omega = \sprod{w* \phi, \nu * \phi}_2 = \sprod{w, \nu}_\calE,
	\end{align*}
so $\nu = \mu$. The second equality is due to the definition of the \emph{dictionary map} $\mu \to \int_{\sph^1} \delta_\theta d\mu(\theta)$, see \cite[p.6]{Flinth2017SoftTV}).
\end{proof}

Next, we need to investigate the continuity properties of the measurement map $M: \calE \to \herm(m)$. It turns out that this is only secured under certain additional conditions on the filter $\phi$. These special conditions will be dealt with subsequently.

\begin{lem} \label{lem:E}
\begin{enumerate}
	\item If the filter $\phi$ has the property that for each $k, \ell$, the sequence $(\widehat{\epsilon}_{\Delta_k - \Delta_\ell}(n) \widehat{\phi}(n)^{-1})_{n \in \Z} $ lies in the sequence space $\ell^2(\Z)$,  then the map $M$ is continuous. ($\widehat{\phi}(n)$ denotes the $n$:th Fourier coefficient of $\phi$) 
	
	\item Under these conditions, the adjoint $M^*$ satisfies
	\begin{align*}
		F_a M^* p = \sum_{k,\ell} p_{k,\ell} \epsilon_{\Delta_k - \Delta_\ell} (\theta)
	\end{align*}
	\end{enumerate}
\end{lem}

\begin{proof}
	$1.$ We have for $k,\ell$ arbitrary
	\begin{align*}
		M_{kl}\mu = \int_{\sph^1} a_k(\theta) \overline{a_\ell(\theta)} d\mu(\theta) = \int_{\sph^1} \epsilon_{\Delta_\ell - \Delta_k} (\theta) d\mu(\theta).
	\end{align*}
	Now, since $\epsilon_{\Delta_\ell - \Delta_k}$ is in $\calC^\infty(\sph^1)$, $\epsilon_{\delta_\ell - \delta_k}(\theta) = \sum_{n \in \Z} \widehat{\epsilon}_{\delta_\ell - \delta_k}(n) \exp(i n \theta)$, where the infinite series converges in $\calC(\sph^1)$. This implies
	\begin{align*}
		\int_{\sph^1} \epsilon_{\Delta_\ell - \Delta_k} (\theta) d\mu(\theta) &= \sum_{n \in \Z} \int_{\T}  \widehat{\epsilon}_{\Delta_\ell - \Delta_k}(n) \exp(i n \omega) d\mu(\omega) = \sum_{n \in \Z} \widehat{\epsilon}_{\Delta_\ell - \Delta_k}(n) \widehat{\mu}(n) \\
		&= \sum_{n \in \Z} \widehat{\epsilon}_{\Delta_\ell - \Delta_k}(n)\hatphi(n)^{-1} \hatphi(n) \widehat{\mu} \leq \left(\sum_{n \in \Z} \abs{\widehat{\epsilon}_{\Delta_\ell - \Delta_k}(n)\hatphi(n)^{-1}}^2  \right)^{1/2}\left(\sum_{n \in \Z} \abs{\hatphi(n) \widehat{\mu}(n)}^2\right)^{1/2}
	\end{align*}
	We applied Cauchy-Schwarz. Now we mearly have to argue that $\sum_{n\in \Z}\abs{\hatphi(n) \widehat{\mu}(n)}^2 = \norm{\mu * \phi}_2^2 = \norm{\mu}_\calE^2$ - but this is the Fourier convolution theorem.
	
	$2.$ We have for $p \in \C^{m,m}$ and $\mu \in \calM(I)$ arbitrary
	\begin{align*}
		\sprod{F_a M^* p, \mu} &= \sprod{(M^*p) * \phi *\overline{\phi}, \mu} = \sprod{M^*p * \phi, \mu *\phi} = \sprod{M^*p, \mu}_\calE = \sprod{p, M \mu} \\
		&=  \overline{\int_{\sph^1} \sum_{k,\ell} \overline{p_{k,\ell}} \epsilon_{\Delta_\ell - \Delta_k} (\theta) d\mu(\theta)}  = \sprod{\sum_{k,\ell} p_{k,\ell} \epsilon_{\Delta_k - \Delta_\ell} (\theta) , \mu},
	\end{align*}
	which proves the claim.
\end{proof}

The subtle task of constructing a filter $\phi$ with desirable properties is postponed to Section \ref{sec:filter}. At this point, let us instead note that the last lemma tells us that the condition $\nu \in \ran F_a M^*$ on the certificates amounts to it being of the form
\begin{align*}
	\sum_{k,\ell} p_{k,\ell} \epsilon_{\Delta_k-\Delta_\ell}.
\end{align*}
The main idea in the following will be to construct a certificate obeying conditions \eqref{eq:AnkareSpecial1}-\eqref{eq:orthCompSameSubSpecial3} by approximating the function $a(\angle(\cdot ,\theta_0))$ with a function of the above form. Since $a(\angle(\cdot,\theta_0))$ satisfies the conditions \eqref{eq:AnkareSpecial1}-\eqref{eq:orthCompSameSubSpecial3} with $t=\sigma=1$, the approximation will also (for less optimal, but still acceptable, values of $t$ and $\sigma$). Therefore, before making the concrete construction of a filter, as well as the soft certificate, we will investigate the approximation properties of the system of \emph{plane waves} $(\epsilon_{\Delta_k-\Delta_\ell})$.

\subsubsection{Approximation with Plane Waves}
In this section, we will for notational purposes not work with the antenna positions $(\Delta_k)$, but rather make claims about general systems of points $(q_i)_{i=1}^N \sse \R^2$ and the corresponding set of plane waves $(\epsilon_{q_i})_{i=1}^N$. The main result will be the following.

\begin{prop} \label{prop:PlaneWaveAppr}
	Let $(q_k)_{k=1}^N$ be a set with an associated covering  $(I_k)_{k =1}^N$ having quality parameters $(\beta(R), \gamma(R))$. Furthermore let $a \in \calC^k(\sph^1)$. Then there exists constants $p_{i}$  with the property
	\begin{align*}
		\sup_{\theta \in \sph^1} \abss{ a(\theta)- \sum_{k=1}^N p_k \epsilon_{q_k}} \leq C \norm{a}_\infty \cdot \gamma(R) + D R^{-k},
	\end{align*}
	where $C$, $D$ are universal constants.
	
	The vector $p$ furthermore obeys $\norm{p}_2 \leqsim \norm{a}_\infty \beta(R)$.
\end{prop}

Before getting into details about the proof, let us sketch the idea of it. Put very informally, we will utilize that every well-behaving function $F$ on $\R^2$ can at least symbolically be written as an ''infinitely long'' linear combination of planar waves through the Fourier transform:
\begin{align*}
	F(x) = \int_{\R^d} \widehat{F}(\xi) \exp(i x \cdot \xi) d\xi
\end{align*} 
The latter integral can furthermore be approximated with a Riemann sum
 \begin{align}
\int_{\R^d} \widehat{F}(\xi) \exp(i x \cdot \xi) d\xi \approx \sum_{i=1}^N \widehat{F}(q_i) \exp(i x \cdot q_i) \abs{I_i}.  \label{eq:RiemannTrick}
\end{align} 
Notice that the right-hand side of the above equation, when evaluated for $x \in \sph^1$, is exactly an expansion in the system $(\epsilon_{q_i})_{i=1}^N$.

As most readers probably have noted, we have yet to touch upon the approximation of $a$, which is a function defined on $\sph^1$, and not necessarily on $\R^2$.  We can however prolong $a$: 
Let $\Phi: [0, \infty) \to \C$ be a Schwarz function with
\begin{itemize}
	\item $\Phi(x) = 0$ for $x \leq \delta$ for some $\delta>0$.
	\item $\Phi(1) =1$.
\end{itemize}
Then define we can prolong $a$ to a function $F: \R^2 \to \C$. as follows:
\begin{align}
	F(x) = \Phi(\abs{x}) a\left(\frac{x}{\abs{x}}\right) \label{eq:fromAToF}
\end{align}
Then a bound on the approximation error \eqref{eq:RiemannTrick}, uniform in $\theta \in \sph^1$ exactly corresponds to a uniform approximation bound for $a$.

Now, for the approximation error \eqref{eq:RiemannTrick} to be small, $\widehat{F}$ needs to have sufficently good decay and smoothness properties. If $a$ is sufficently smooth, this is the case, as is shown in the following lemma.

\begin{lem} \label{lem:Fproperties}
	Let $a \in \calC^k(\sph^1)$. The function \eqref{eq:fromAToF} then satisfies the decay estimate $$\abs{\widehat{F}(s\omega)}, \leqsim s^{-(k+2)}.$$
	We can furthermore uniformly bound $\widehat{F}$ as well as the derivatives of $\widehat{F}$ of first and second order as follows:
	\begin{align*}
		\norm{\widehat{F}^{k}}_\infty \leq \norm{a}_\infty \norm{\partial^{k+1}\Phi}_\infty \abs{\sph^1}, \quad k=0, 1, 2.
	\end{align*}
\end{lem}
\begin{proof}
Throughout the proof, we will use polar coordinates. Letters $s$ and $r$ will be used to denote radii, and $\eta$ and $\theta$ for the directions in $\sph^1$.

\emph{Decay:} First, by the definition of the Fourier transform, we have
\begin{align}
	\widehat{F}(s\theta) &=  \int_{\sph^1} a(\eta) \int_{0}^\infty \Phi(r)r \exp(irs \sprod{\eta, \theta}) d\eta dr \nonumber \\
	& =  \int_{\sph^{1}} a(\eta) \widehat{(r \cdot \Phi(r))}(s\sprod{\eta, \theta}) d\eta = \int_{\sph^{1}} a(\eta) i \partial_{1}\widehat{ \Phi}(s\sprod{\eta, \theta}) d\eta. \label{eq:IntDarstellung}
\end{align}
Let us drop the imaginary unit in the rest of the calculations, since it does not affect the decay. Using the identification $\sph^1 \simeq \T$, and the substitution $\angle(\eta,\theta) = \arccos(t)$, the latter integral is equal to
	\begin{align*}
	\int_{\sph} a(\eta) \partial^1 \widehat{\Phi}(s \cos(\angle(\eta,\theta))) d\theta &= \int_{-1}^1 \frac{a(\omega - \arccos(t))\partial^1 \widehat{\Phi}(st)}{\sqrt{1-t^2}} dt+ \int_{-1}^{1} \frac{a(\omega + \pi- \arccos(t))\partial^1 \widehat{\Phi}(st)}{\sqrt{1-t^2}} dt \\
		&= \int_{-1}^1 \frac{a(\omega - \tfrac{\pi}{2} +\arcsin(t))\partial^1 \widehat{\Phi}(st)}{\sqrt{1-t^2}} dt+ \int_{-1}^{1} \frac{a(\omega + \tfrac{\pi}{2}+ \arcsin(t))\partial^1 \widehat{\Phi}(st)}{\sqrt{1-t^2}}.
	\end{align*}
	 In the following, we will concentrate on the first integral in the final sum, since the second can be treated with similar methods. Since $a$ is continuous, there exists a function $A: (-\pi, \pi) \to \C$ with $A'(\theta) = a(\theta)$ and $A(\omega- \tfrac{\pi}{2}) = 0$. Integrating by parts yields that the first integral equals
	\begin{align*}
		s^{-1}A(\omega - \tfrac{\pi}{2} +\arcsin(t))\partial^2 \widehat{\Phi}(st)\big\vert_{-1}^1-s^{-1}\int_{-1}^1 A(\omega - \tfrac{\pi}{2} +\arcsin(t))\partial^2 \widehat{\Phi}(st) dt.
	\end{align*}
	The first one of these terms decay rapidly, since $\partial^2 \widehat{\Phi}(s)$ does. As for the second, we perform a Taylor expansion for $A$ around $\omega_0=\omega - \tfrac{\pi}{2}$.
	\begin{align*}
		\int_{-1}^1 A(\omega - \tfrac{\pi}{2} +\arcsin(t))\partial^2 \widehat{\Phi}(st) dt = \int_{-1}^1 \left(\sum_{\ell=0}^{k-1} a^{(\ell)}(\omega_0) \arcsin(t)^{\ell+1}  + g(\arcsin(t)) \right)\partial^2 \widehat{\Phi}(st) dt
	\end{align*}
	with $g(r) \leqsim \norm{a^{(k)}}_\infty r^{k+1}$. A Taylor expansion of $\arcsin(x)$ at zero yields
	\begin{align*}
		\arcsin(x) = \sum_{n=0}^\infty \frac{1}{2^{2n}} \binom{2n}{n} \frac{x^{2n+1}}{2n+1}, \abs{x} \leq 1.
	\end{align*}
This implies in particular for any $M \in \N$.
\begin{align*}
	\abss{\arcsin(x) - \sum_{n=0}^M\frac{1}{2^{2n}}\binom{2n}{n} \frac{x^{2n+1}}{2n+1}} &\leq \abs{x}^{2M+3} \sum_{n=M+1}^\infty\frac{1}{2^{2n}}\binom{2n}{n} \frac{\abs{x}^{2n+1-2M+1}}{2n+1} \leq \abs{x}^{2M+3}\sum_{n=0}^\infty \frac{1}{2^{2n}}\binom{2n}{n} \frac{x^{2n+1}}{2n+1} \\
	&= \abs{x}^{2M +3} \arcsin(1) = \frac{\pi}{2}\abs{x}^{2M+3},
\end{align*}
which implies
\begin{align*}
	 \int_{-1}^1 \partial^2 \widehat{\Phi}(st) \arcsin(t) dt = \int_{-1}^1 \partial^2 \widehat{\Phi}(st) \sum_{n=0}^M\frac{1}{2^{2n}}\binom{2n}{n} \frac{t^{2n+1}}{2n+1}dt  + \int_{-1}^1 \partial^2 \widehat{\Phi}(st) h(t) dt
\end{align*}
where $\abs{h(t)} \leq \pi \abs{t}^{2M+3}/2$. The latter term can be estimated
\begin{align*}
	 \abss{\int_{-1}^1 \partial^2 \widehat{\Phi}(st) h(t) dt} \leq \frac{\pi}{2} \int_{-1}^1 \abs{\partial^2 \widehat{\Phi}(st)} \abs{t}^{2M+3} dt =  \frac{\pi}{2} s^{-(2M+4)}\int_{-s}^s \abs{\partial^2 \widehat{\Phi}(u)} \abs{u}^{2M+3} du \leqsim s^{-(2M+4)},
\end{align*}
	since the latter integral converges for $s \to \infty$. As for the other term, it suffices to prove that  $\int_{-1}^1 t^m \partial^2 \widehat{\Phi}(st) dt$ decays rapidly for $s \to \infty$, where the implicit constants only depends on $m$. To see that this is the case, we integrate by parts twice to obtain
	\begin{align*}
		\int_{-1}^1 t^m \partial^2 \widehat{\Phi}(st) dt = s^{-1} \partial^1\widehat{\Phi}(st)t^m   - m s^{-2} \widehat{\Phi}(st)t^{m-1} \big\vert_{-1}^1 + m(m-1)s^{-3}\int_{-1}^1 t^{m-2} \widehat{\Phi}(st) dt  
	\end{align*}
	The first two of these terms again decay rapidly. As for the last, we have
	\begin{align*}
		\int_{-1}^1 t^{m-2} \widehat{\Phi}(st) dt  = s^{-(m-1)} \int_{-s}^s u^{m-2} \widehat{\Phi}(u) du =  s^{-(m-1)} \int_{-\infty}^\infty u^{m-2} \widehat{\Phi}(u) du  -s^{-(m-1)} \int_{\abs{u} \geq s } u^{m-2} \widehat{\Phi}(u) du
	\end{align*}
	The first of the last two integrals is equal to a multiple of $\partial^{m-2} \Phi(0)$, and hence vanishes, since $\Phi$ does in a neighborhood around zero. The second decays rapidly due to the rapid decay of $\widehat{\Phi}$. 
	
	\emph{Uniform Bound:} The integral representation \eqref{eq:IntDarstellung} together with the Hölder inequality already implies the bound on $\widehat{F}$. Since
	\begin{align*}
	\nabla \widehat{F}(s\theta) &= i \int_{\sph^1} a(\eta) \eta \partial^{2} \widehat{\Phi}(s \sprod{\eta, \theta}) d\eta \\	
	\widehat{F}''(s\omega) & = i\int_{\sph^1} a(\eta) \eta\eta^* \partial^{3} \widehat{\Phi}(s \sprod{\eta, \theta}) d\eta, 
\end{align*}
the other two bounds follow similarly.

\end{proof}

With the last lemma in our toolbox, we can prove Proposition \ref{prop:PlaneWaveAppr}.

\begin{proof}[Proof of Proposition \ref{prop:PlaneWaveAppr}]
Using the notation from above, we define $p_k = \widehat{F}(p_k)\abs{I_k \cap B_R(0)}$. According to the above discussion, we then have
\begin{align*}
	a(\theta) - \sum_{k=1}^N p_k \epsilon_{q_k}(\theta) = \int_{\R^d} \widehat{F}(\xi) \exp(i \xi \cdot \theta) d\xi - \sum_{k=1}^N \widehat{F}(p_k)\exp( i \sprod{\theta,q_k}) \abs{I_k \cap B_R(0)}.
\end{align*} 
In order to estimate this expression, note that the $R$-covering property implies
\begin{align*}
&\abss{\int_{\R^d} \widehat{F}(\xi) \exp(i \xi \cdot \theta) d\xi - \sum_{k=1}^N \widehat{F}(p_k) \exp(i x \cdot p_k) \abs{I_k}} \leq \underbrace{\int_{\abs{\xi} \geq R} \abs{\widehat{F}(\xi)} d\xi}_{{\bf A}} \\
 &\qquad \qquad \qquad \qquad \qquad+ \underbrace{\abss{\sum_{k=1}^N \int_{B_R(0) \cap I_k} \widehat{F}(\xi) \exp(i \xi \cdot x) - \widehat{F}(p_k) \exp(i p_k \cdot x) d\xi}}_{\bf B}.
	\end{align*}
The term ${\bf A}$ is easily bounded with the help of decay estimate from Lemma \ref{lem:Fproperties}:
		\begin{align*}
			\int_{\abs{\xi} \geq R} \abss{\widehat{F}(\xi)} d\xi \leqsim \abss{\sph^1} \int_{R}^\infty s^{-(k+2)} s ds \leqsim R^{-k}.
\end{align*}
As for $\bf B$, we have to work a bit more. Define the functions $G_x: \R^2 \to \C$ by $G_x(\xi) =\widehat{F}(\xi)\exp(i x \cdot \xi)$. For the integral over a domain $B_R(0)\cap I_k$, we perform a Taylor expansion in $q_k$ to arrive at
	\begin{align*}
	&\abss{\int_{B_R(0) \cap I_k} \widehat{F}(\xi) \exp(i \xi \cdot \theta) - \widehat{F}(q_k) \exp(i q_k \cdot x) d\xi } =	\abss{\int_{B_R(0) \cap I_k}  G_x(\xi)-G_x(q_k) d\xi }   \\
			 & \qquad \qquad \leq   \abss{ \int_{B_R(0) \cap I_k}  \nabla G_x (p_i) (\xi-q_k)   d\xi} + \frac{1}{2}\int_{B_R(0) \cap I_i} \sup_{\upsilon \in I_i \cap B_R(0)} \norm{ G''_x(\upsilon)} \abs{\xi-q_k}^2 d\xi.
		\end{align*}
Due to the centroid property, the first of these terms vanish. As for the second, we argue that due to $G_x'' = -xx^* F +2i x \nabla F + F$, the uniform bounds in Lemma \ref{lem:Fproperties} imply
\begin{align*}
	\sup_{x \in \sph^1, \upsilon \in \R^2} \norm{ G''_x(\upsilon)} \leqsim \norm{a}_\infty\left(\sum_{j=1}^3 \normm{\partial^{j} \widehat{\Phi}}_\infty \right).
\end{align*}
 Therefore
 \begin{align*}
 	\abss{\int_{B_R(0) \cap I_k} \sup_{\upsilon \in I_k \cap B_R(0)} \norm{ G''_x(\upsilon)} \abs{\xi-q_k}^2 d\xi }\leqsim \norm{a}_\infty  \abs{B_R(0) \cap I_k} \diam(I_k \cap B_R(0))^2.
 \end{align*}
 The first claim  now follows by summation over $k$.

	As for the bound on the norm of $p$, we again apply the uniform bound from Lemma \ref{lem:Fproperties} to obtain
	\begin{align*}
		\norm{p}_2 = \sum_{k=1}^N \abs{\widehat{F}(q_k)}^2 \cdot \abs{ I_k \cap B_R(0)}^2 \leqsim \norm{a}_{\infty}^2\max_k \abs{I_k \cap B_R(0)} \cdot \sum_{k=1}^N  \abs{ I_k \cap B_R(0)} = \norm{a}_{\infty}^2\beta(R) \cdot \pi R^2
	\end{align*}
	where the last inequality follows from the $R$-covering property. 
\end{proof}

\subsubsection{Proof of Theorem \ref{th:main}}

We can now prove the main theorem.

\begin{proof}[Proof of Theorem \ref{th:main}] Lemma \ref{lem:E} proves that  $\ran F_a M^*$ is spanned by the plane waves $(\epsilon_{\Delta_\ell - \Delta_k})$. Therefore, applying Proposition \ref{prop:PlaneWaveAppr} to the function $a(\cdot - \theta_0)$, where $a$ is the autocorrelation function associated to the function $\phi$ yields a function $g \in \ran F_a M^*$ with
\begin{align}
	\sup_{\theta \in \sph^1} \abs{ g(\theta) - a(\angle(\theta))} \leq  \abs{g(\theta_0)} (K \gamma(R) + CR^{-K}), \label{eq:certbound}
\end{align}
where we leave the choice of $g(\theta_0)$ open for now. Proposition \ref{prop:PlaneWaveAppr} also proves that the corresponding $p$-vector obeys $\norm{p}_2 \leqsim \abs{g(\theta_0)}\beta(R)$.

Define $\lambda := \sup_{\theta \in \supp \mu_c} \abs{a(\angle(\theta,\theta_0))}$. Choosing the sign of $g(\theta_0)$ to be the one conjugate to $c_{\theta_0}$ and utilizing \eqref{eq:certbound}, we obtain
\begin{align*}
	\re\left(\int_{\sph^1} g(\theta)d\mu(\theta) \right) &= \re\left ( \int_{\sph^1} g(\theta_0) a(\angle(\theta, \theta_0)) d\mu(\theta) \right)-\abs{\sph^1} (K\gamma(R) + CR^{-K})\abs{g(\theta_0)} \\
	&\geq \re\left ( g(\theta_0) c_{\theta_0} a(\angle(\theta_0, \theta_0)) + \int_{\sph^1} g(\theta_0)a(\angle(\theta,\theta_0)) d\mu_c(\theta) \right)-\abs{\sph^1} (K\gamma(R) + CR^{-K})\abs{g(\theta_0)}\\
	&\geq \underbrace{\abs{g(\theta_0)}\left( \abs{c_{\theta_0}} - \lambda(1-c_{\theta_0}) - \abs{\sph^1}( K \gamma(R) + CR^{-k})\right)}_{(*)}
\end{align*}
where we used the separation condition and the normalization assumption. With the choice
\begin{align*}
	 \abs{g(\theta_0)} =\left( \abs{c_{\theta_0}} - \lambda(1-c_{\theta_0}) - \abs{\sph^1}( K \gamma(R) + CR^{-k})\right)^{-1},
\end{align*}
$(*)$ is equal to $1$, so that \eqref{eq:AnkareSpecial1} is fulfilled. Identifying $\sigma$ with $\abs{g(\theta_0)}$ and utilizing \eqref{eq:certbound}, we see that \eqref{eq:atPointSpecial2} and \eqref{eq:orthCompSameSubSpecial3} are satisfied with
\begin{align*}
\sigma = \left(\abs{c_{\theta_0}}( 1+\lambda) - \lambda - \abs{\sph^1}(K\gamma(R) + CR^{-k})) \right)^{-1}\\
	 t = 1- \frac{ \abs{\sph^1}K\gamma(R) + CR^{-k}}{\abs{c_{\theta_0}}( 1+\lambda) - \lambda - \abs{\sph^1}(K\gamma(R) + CR^{-k}))}.
\end{align*}
Remembering that $\lambda \leq \tfrac{\abs{c_{\theta_0}}}{6}$ and using the assumption $\abs{\sph}^1(K(\gamma(R) + C R^{-k}< \tfrac{\abs{c_{\theta_0}}}{6}$, we may estimate
\begin{align*}
	\sigma& \leq \abs{c_{\theta_0}}^{-1} \left( 1 - \frac{2}{6}\right)^{-1} = \frac{3}{2} \abs{c_{\theta_0}}^{-1}, \sigma \geq \abs{c_{\theta_0}^{-1}}\\
	\frac{t}{\sigma} &= \abs{c_{\theta_0}}( 1+\lambda) - 2\lambda - 2\abs{\sph^1}(K\gamma(R) + CR^{-k}))  \geq \abs{c_{\theta_0}} \left( 1 - \frac{4}{6}\right)= \frac{\abs{c_{\theta_0}}}{3}
\end{align*}
The first claim of the theorem now directly follows from Theorem ... As for the second, we apply the same theorem, remember that $\norm{p}_2 \leq \sigma \beta(R)$ and estimate
\begin{align*}
	\frac{2\norm{p}_2\overline{e}+ (\rho-1)}{\rho \sigma} \leq \frac{3 \beta(R) \abs{c_{\theta_0}}^{-1}\overline{e} + (\rho-1) }{\rho \tfrac{3}{2}  \abs{c_{\theta_0}}^{-1}},
\end{align*}
from which the claim follows.
\end{proof}

\subsubsection{Construction of a Suitable Filter.} \label{sec:filter}
Before getting in to the construction of the filter $\phi$, let us quickly summarize which properties  it, or equivalently $a$, ideally should possess:
\begin{enumerate}[(i)]
		\item $a$ should be as smooth as possible. \dots The exact reasons for this will become evident in the next session.
		\item $a$ should decay quickly when moving away from $0$ -- if not, the statement $\abs{a(\theta-\theta_0)}$ will not mean much.
		\item $\phi$ should not be too smooth, or rather:  $(\widehat{\phi}(n))_{n \in \Z}$ should not decay to quickly, since this would prevent $(\widehat{\epsilon}_{\Delta_k - \Delta_\ell}(n) \widehat{\phi}(n)^{-1})_{n \in \Z}$ to be in $\ell^2(\Z)$.
		\item $\phi$ should be $L^2$-normalized.
	\end{enumerate}
	
	The strategy for constructing such a filter will be to modify a filter whose correlation function is the \emph{Fej\'{e}r} kernel $\alpha_M$. For $M>0$, the  Fej\'{e}r kernel is defined as
	\begin{align*}
		\alpha_M(\omega)= \frac{1}{M^2}\cdot \frac{1-\cos(M\cdot\theta)}{1-\cos(\theta)} = \sum_{\abs{n} \leq M} \left(1 - \frac{\abs{n}}{M}\right) e^{in \theta}
	\end{align*}
	for $\omega \in (0,2\pi)$, and $\alpha(0)=1$. Excluding the anti-smoothness condition $(iii)$, $\alpha_M$ has all the desired properties. $(i)$ and $(iv)$ are imminent, and as for $(ii)$, a Taylor expansion as well as the bound $\abs{1-\cos(x)}\leq 2$ reveals
	\begin{align}
		\alpha(\omega)  \leq \begin{cases} 1-\frac{M^2\omega}{12} + \frac{M^4\omega^4}{360},&  \abs{\omega} \leq \frac{2\pi}{M}\\
		\frac{2}{1-\cos\left(\frac{2\pi\lfloor M\omega\rfloor}{M} \right)}, & \abs{\omega} > \frac{2\pi}{M} \end{cases} . \label{eq:fejerdecay}
	\end{align}

\begin{figure}
\centering
	\includegraphics[width=7cm]{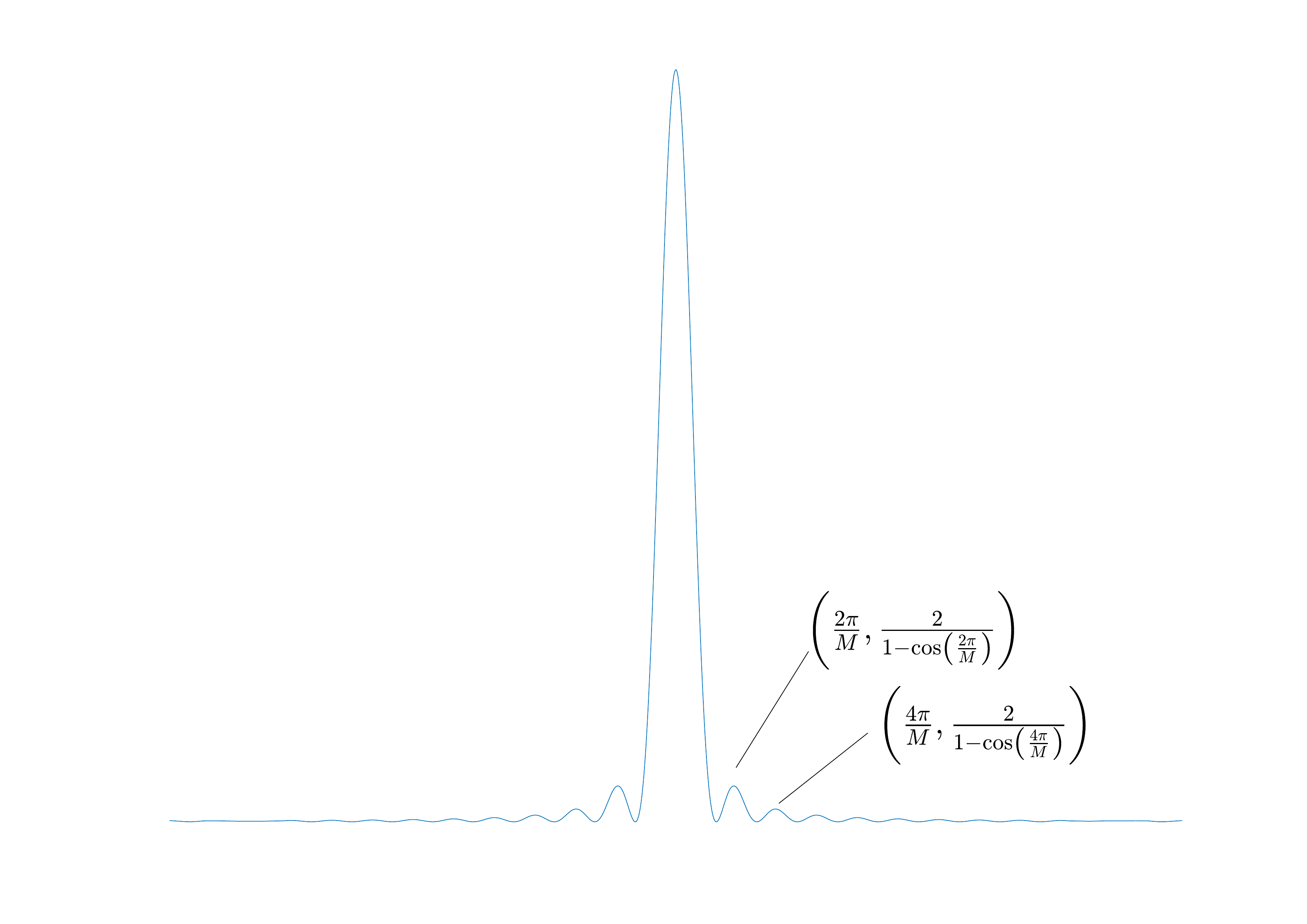}
	\caption{The Fej\'{e}r kernel.}
\end{figure}

We now construct the filter. The idea will be to construct a function whose autocorrelation almost equals $\alpha_M$, but in contrast does satisfy $(iii)$.

\begin{prop} \label{prop:filter}
	Let $k$ and $M$ be positive natural numbers,with $M \geq 3$. There exists an $L^2$-normalized filter filter $\phi$ with autocorrelation function $a = \phi *\phi$ obeying
	\begin{align}
		a &\in \calC^{2k-2} \label{eq:Asmooth}\\
		a(\omega) &\leq \begin{cases} 1-\frac{M^2\omega^2}{12} + \frac{M^4\omega^4}{360} + \frac{C}{M^{2k}},&  \abs{\omega} \leq \frac{2\pi}{M} \label{eq:ADecay}\\
		\frac{2}{1-\cos\left(\frac{2\pi\lfloor M\omega\rfloor}{M} \right)}+ \frac{C}{M^{2k}}, & \abs{\omega} > \frac{2\pi}{M} \end{cases} \\
		\abs{ \widehat{\phi}(n)} &\geq DM^k \abs{n}^{-k}  \label{eq:antiDecay}
	\end{align}
	for constants $C, D>0$.
\end{prop}

\begin{proof}
	Consider the filter
\begin{align*}
		\phi_M(\omega)= \frac{1}{\left(1 + \sum_{\abs{n}> M} \frac{1}{n^{2k}}\right)^{1/2}}\left(  \sum_{\abs{n} \leq M} \sqrt{1 - \frac{\abs{n}}{M}} e^{in \omega} + \sum_{\abs{n} > M} \frac{1}{\abs{n}^k} e^{in \omega}\right).
	\end{align*}
	Then $\phi_M$ obeys the anti-decay property \eqref{eq:antiDecay}, since
	\begin{align*}
		\abs{n}^k\abs{\widehat{\phi}(n)} &= \frac{1}{\left(1 + \sum_{\abs{n}> M} n^{-2k}\right)^{1/2}} \begin{cases}	M^k \left(\frac{\abs{n}}{M}\right)^k \sqrt{1- \frac{\abs{n}}{M}} ,& \abs{N}\leq M\\
		1 ,& \abs{N}>M
		\end{cases} \\
		&\leq \frac{1}{\sqrt{1+ \frac{\pi^2}{3}}}\cdot\max\left(\frac{M^ke^{-1}}{\sqrt{2k}},1 \right)
		 \leq \frac{1}{\sqrt{1+ \frac{\pi^2}{3}}}\cdot\max\left(\frac{M^k}{2e},1\right) = \frac{1}{2e\sqrt{1+ \frac{\pi^2}{3}}}M^k.
	\end{align*} 
	The Fourier convolution theorem implies that the corresponding filter $a_M$ takes the form 
	\begin{align*}
		a_M(\omega) = \phi_M * \phi_M(\omega)=  \frac{1}{1 + \sum_{\abs{n}> M} n^{-2k}}\left( \alpha_M(\omega) + \sum_{\abs{n} > M} \frac{1}{\abs{n}^{2k}} e^{in \omega}\right).
	\end{align*}
	We then have $a_M(0)=1$, which directly corresponds the normalization of $\phi$. To prove \eqref{eq:Asmooth}, note that the symbolic $(2k-2)$:th derivative is a Fourier series where the coefficients are in $\ell_1(\Z)$. Therefore, by standard arguments, the symbolic series converges almost everywhere to a continuous function, which then is the $2k-2$:th derivative of $a_M$.
	
	Finally, to prove \eqref{eq:ADecay}, first notice that
	\begin{align*}
		\sup_{\omega}& \abss{\frac{1}{1 + \sum_{\abs{n}> M} n^{-2k}}\left( \alpha_M(\omega) + \sum_{\abs{n} \geq M} \frac{1}{n^{2k}} e^{in \omega}\right) - \alpha_M(\omega)} \\
	&\leq \left(\frac{1}{1 + \sum_{\abs{n}> M} n^{-2k}}-1\right) \sup_\omega \abs{\alpha_M(\omega)} + \frac{ \sum_{\abs{n}> M} n^{-2k}}{1 + \sum_{\abs{n} > M} n^{-2k}} \\
	& \leq \left( 1+ \frac{1}{1 +\frac{\pi^2}{3}} \right) \sum_{\abs{n}> M} n^{-2k} \leqsim M^{-2k}.
	\end{align*}
		\eqref{eq:ADecay} now easily follows from the decay estimate \eqref{eq:fejerdecay}.
\end{proof}

\subsection{Bounding $\beta(R)$ and $\gamma(R)$ for Circular Designs}\label{sec:CircDesign}

Let us begin by explicitly calculating the difference set associated with the  circular design.

\begin{lem}
	Let $m$ be even. The difference set $(\Delta_k - \Delta_\ell)_{k,\ell = 1}^m$  for the circular design \eqref{eq:circular} is given by
	\begin{align*}
\set{ 2\rho\sin(\tfrac{j\pi}{m}) e^{\tfrac{ij\pi}{m}}e^{\tfrac{2\pi i \kappa}{m}} \ \vert \ \kappa= 1,\dots, m, \ j = 0, \dots, \frac{m}{2} } 
\end{align*}

\end{lem}

\begin{proof}
For $k$ and $\ell$ arbitrary, we have
\begin{align*}
	\Delta_k - \Delta_\ell &= \rho \left(e^{\tfrac{2\pi i k}{m}}- e^{\tfrac{2\pi i \ell}{m}}\right) = \rho e^{\tfrac{2\pi i k}{m}} \left( 1 - e^{\tfrac{2\pi i (\ell-k)}{m}} \right) = \rho e^{\tfrac{2\pi i k}{m}} \left( 1 - \cos\left(\tfrac{2\pi  (\ell-k)}{m}\right) -i \sin\left(\tfrac{2\pi  (\ell-k)}{m}\right) \right) \\
	&= \rho e^{\tfrac{2\pi i k}{m}} \left( 2 \sin^2\left(\tfrac{\pi  (\ell-k)}{m}\right)  -2i \sin\left(\tfrac{\pi  (\ell-k)}{m}\right) \cos\left(\tfrac{\pi  (\ell-k)}{m}\right) \right) \\
	&= 2\rho e^{\frac{2\pi i k}{m}} \sin\left(\tfrac{\pi  (\ell-k)}{m}\right)\left(-i \cos\left( \tfrac{\pi(\ell-k)}{m}\right) + \sin\left( \tfrac{\pi(\ell-k)}{m}\right)  \right) =  -i 2\rho e^{\frac{2\pi i k}{m}} \sin\left(\tfrac{\pi  (\ell-k)}{m}\right)e^{\tfrac{i\pi  (\ell-k)}{m}}.
\end{align*}
Let us  rename $j= k-\ell$. The set of values $k$ for which there exists an $\ell \in {1, \dots m}$ such that $k-\ell=j$ is given by $\set{-(j-1), \dots (m-j)}$, an interval of length $(m-1)$. Hence, due to the periodicity of the complex exponential , $e^{\frac{2\pi i k}{m}}$ takes on all of the values $e^{\frac{2\pi i \lambda}{m}}, \ \lambda= 1,\dots, m$ for each value of $j$. By shifting $\lambda$ with $m/2$, we may further cancel out the $-i$-constant, leaving us with a set
\begin{align*}
	S_j =\set{ 2\rho\sin(\tfrac{j\pi}{m}) e^{\tfrac{ij\pi}{m}}e^{\tfrac{2\pi i \lambda}{m}} \ \vert \ \lambda= 1,\dots, m } 
\end{align*}
for a fixed value of $j$. Now let us argue that we only need to consider values of $j$ between $0$ and $m/2$. First, due to the antisymmetry of $\sin$, $$\sin(\tfrac{-j\pi}{m}) e^{\tfrac{-ij\pi}{m}} e^{\tfrac{2\pi i \lambda}{m}} = \sin(\tfrac{j\pi}{m}) e^{\tfrac{ij\pi}{m}}e^{\tfrac{2\pi i (\lambda+m/2-j)}{m}}.$$ Again by shifting $\lambda$, we see that $S_j= S_{-j}$, so that we only need to consider non-negative values of $j$. The symmetry $\sin(\pi-x)=\sin(x)$ finally implies
 $$\sin(\tfrac{(m-j)\pi}{m}) e^{\tfrac{i(m-j)\pi}{m}} e^{\tfrac{2\pi i \lambda}{m}} =- \sin(\tfrac{j\pi}{m}) e^{\tfrac{-ij\pi}{m}}e^{\tfrac{2\pi i \lambda}{m}} = \sin(\tfrac{-j\pi}{m}) e^{\tfrac{-ij\pi}{m}}e^{\tfrac{2\pi i \lambda}{m}},$$
 so that $S_{m-j} = S_{j}$. Hence, we in fact only need to consider values of $j$ between $0$ and $m/2$, and the proof is finished.

\end{proof}
\begin{figure}
\centering
	\includegraphics[scale=.3]{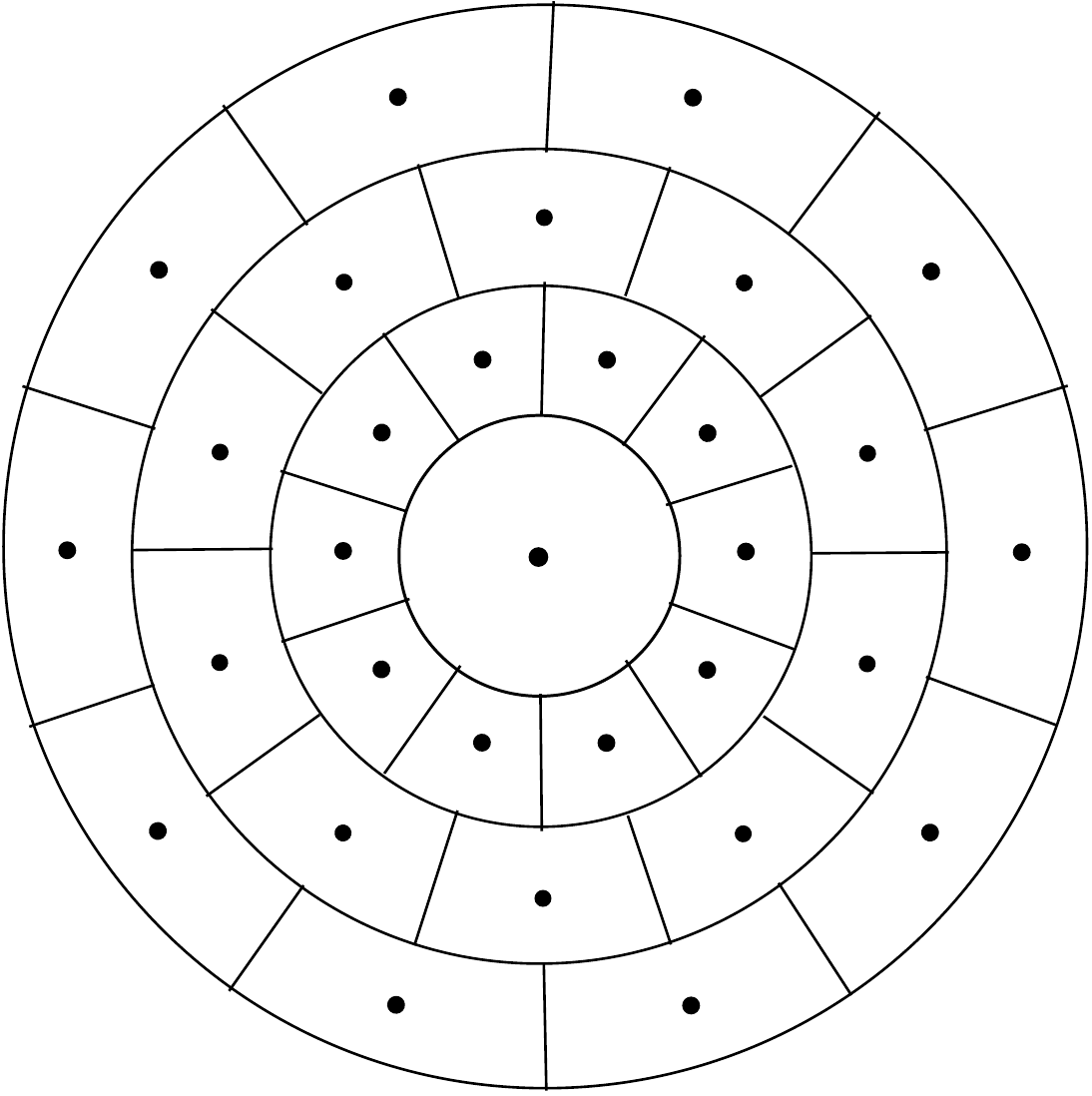}
	\caption{
		The sets $I_{j\ell}$ used for covering the difference set induced by the circular antenna design. \label{fig:Covering}
	}
\end{figure}

Let us for notational simplicity rename the points in the difference set  to $q_{jk}:= 2\rho\sin(\tfrac{j\pi}{m}) e^{\tfrac{ij\pi}{m}}e^{\tfrac{2\pi i k}{m}}$. Our task is now to find a collection of disjoint sets $I_{jk}$ having the $R$-covering property, and additionally that  $q_{jk}$ is the centroid of $I_{jk}$ for each $j,k$. We will do this the following way: let $r_j$ denote a set of radii with 
\begin{align} 0=r_0 \leq r_1 \leq 2\rho\sin(\tfrac{\pi}{m}) \leq r_2 \leq \dots \leq 2\rho\sin(\tfrac{\pi}{2})<r_{\tfrac{m}{2}+1} \label{eq:radii}
\end{align} and define
\begin{align} \label{eq:radialRectangles}
	I_{j\ell} = \set{r e^{i \theta} \ \vert \ r \in (r_{j}, r_{j+1}), \theta \in \frac{\pi(j+2k)}{m}+\left(-\frac{\pi}{m},\frac{\pi}{m}\right)},
\end{align}
and $I_{0} = B_{r_1}(0)$. (A graphical depiction of this covering is given in  Figure \ref{fig:Covering}). 
Then $q_{j\ell} \in I_{j\ell}$, and $I_{j\ell}$ disjointly, up to sets of Lebesgue measure zero, covers $B_{2\rho}(0)$. It is however not a priori clear that we can can choose the radii $r_j$ in such a way so that $q_{j\ell}$ is the centroid of $I_{j\ell}$. The next lemma shows that  it at least is possible to choose the radii so that a large portion of the points $q_{j\ell}$ are the centroids of their respective sets. This will be enough to bound the $\beta(R)$ and $\gamma(R)$-constants of the circular design to a satisfactory extent.

\begin{lem} \label{lem:radii}
	Let $m \geq 5$. There exists a universal constant $\Theta$ such that for any $\rho>0$, there exists a sequence of radii $r_j$ fulfilling \eqref{eq:radii}, while still every $q_{j\ell}$ with $q_{j\ell} \leq \Theta \rho$ is the centroid of $I_{j\ell}$.
\end{lem}

The proof of this proposition is quite technical, so that we postpone it to Section \ref{sec:radii}. Let us instead use it to bound the quality parameters.

\begin{proof}[Proof of Proposition \ref{prop:CircularQualityParameters}]
	Define sets $I_{jl}$ as in \eqref{eq:radialRectangles} for the $q_{j\ell}$ with $\abs{q_{j\ell}} \leq  \Theta \rho$, with radii $r_i$ and $\Theta$ as in Lemma \ref{lem:radii}, and simply $I_{jl}=\set{q_{jl}}$ for the other points. Lemma \ref{lem:radii} then implies that the sets $I_{jl}$ have the $(\Theta \rho)$-covering property. Since
	\begin{align*}
		\abs{I_{jl}} = \frac{2\pi}{m} \cdot \pi (r_{j+1}^2 - r_j^2) , \quad \diam(I_{jl})^2 \leq r_{j+1}^2 \left(\frac{2\pi}{m}\right)^2 + (r_{j+1}-r_j)^2,
\end{align*}
 we can thus estimate
	\begin{align*}
		\beta(\Theta \rho) \leq \Theta \rho \sqrt{\max_{j }\frac{2\pi^2}{m}(r_{j+1}+r_j)(r_{j+1}-r_j)} \leq \Theta \rho \sqrt{\frac{2\pi^2}{m} 2 \cdot 2\rho \cdot 2 \rho \frac{2\pi}{m}} = 4\sqrt{2}\Theta \pi^{\tfrac{3}{2}} \rho^2m^{-1}.
	\end{align*}
	We used the Lipschitz continuity of $\sin$ together with $ r_{j+1}-r_j \leq 2\rho \left( \sin \left(\frac{\pi(j+1)}{m}\right) - \sin\left(\frac{\pi (j-1)}{m} \right)\right)$ in the second step. As for the $\gamma(\Theta \rho)$-constant, we have
	\begin{align*}
\gamma(\Theta \rho) &= \diam(I_0)^2 \abs{I_0} + \sum_{j=1}^{m/2} \sum_{k=1}^m \diam(I_{j\ell})^2 \abs{I_{j\ell}} \\
	&\leq  (2r_1)^2 \pi r_1^2 +\sum_{j=1}^{m/2} \sum_{k=1}^m \left(r_{j+1}^2 \left(\frac{2\pi}{m}\right)^2 + (r_{j+1}-r_j)^2\right) \left(\frac{2\pi}{m}(r_{j+1}^2-r_j^2)\right) \\
	&\leq 4\rho^2 \left(\frac{\pi}{m}\right)^2 \pi 4\rho^2 + \sum_{j=1}^{m/2} \sum_{k=1}^m \left(4\rho^2 \left(\frac{2\pi}{m}\right)^2 + \left(2\rho \frac{2\pi}{m}\right)^2\right) \left(\frac{2\pi}{m} \cdot 4\rho^2 \frac{\pi}{m}\right)  \\
	&\leq \frac{16\pi^3\rho^4}{m^2}+\frac{m^2}{2}\cdot \frac{256\pi^3\rho^4}{m^4} = 144\pi^3 \rho^4 m^{-2}.
	\end{align*}
	Note that we made quite a few very crude estimates in the final chain of estimations -- but that making them less crude will not change the asymptotics. The proof is finished.
\end{proof}

\subsubsection{Proof of Lemma \ref{lem:radii}} \label{sec:radii}

It still remains to construct the sequence of radii \eqref{eq:radii}. Let us begin by calculating the centroid of a set of the form \eqref{eq:radialRectangles}.

\begin{lem} \label{lem:centroids}
The centroid of the domain $I_{j\ell}$ for $\ell \geq 1$ is given by
\begin{align*}
	\frac{r_{j}^2 + 2 r_j r_{j+1} +r_{j+1}^2}{3(r_j + r_{j+1})} \sinc \left( \frac{\pi}{m}\right) e^{\tfrac{2\pi i \ell}{m}}
\end{align*}
\end{lem}

\begin{proof}
 In order to simplify notation, let us denote $r_j =R_1$, $r_{j+1}=R_2$, $\frac{2\pi \ell}{m}=\alpha$ and $\omega = \frac{\pi}{m}$. We  have
	\begin{align*}
		\int_{I_{j\ell}} x_1 dx = \int_ {R_1}^{R_2} \int_{\alpha- \omega}^{\alpha + \omega} (r \cos(\vphi)) r d\vphi dr	 &= \frac{R_2^3-R_1^3}{3} \left( \sin(\alpha + \omega)- \sin(\alpha - \omega) \right) \\
		&= \frac{2}{3} (R_2^3-R_1^3) \cos(\alpha) \sin(\omega)	 \\
		\int_{I_{j\ell}} x_2 dx = \int_ {R_1}^{R_2} \int_{\alpha- \omega}^{\alpha + \omega} (r \sin(\vphi)) r d\vphi dr	 &= \frac{R_2^3-R_1^3}{3} \left( \cos(\alpha + \omega)- \cos(\alpha - \omega) \right) 
	\\
	&= \frac{2}{3} (R_2^3-R_1^3) \sin(\alpha) \sin(\omega).
	\end{align*}
	Since the area $I_{j\ell}$ is given by $\omega (r_2^2-r_1^2)$, the centroid is given by
	\begin{align*}
		\frac{2(r_1^2+r_1r_2+r_2^2)}{3(r_1+r_2)} \sinc(\theta) e^{i\alpha},
	\end{align*}
	which is what to be proven.
	\end{proof}
	
	Lemma \ref{lem:centroids} implies that the centroid of $I_{j\ell}$ for any choice of $r_j$ is collinear with the points $q_{j\ell}$. The only question left to answer is if it is possible to assign values $r_j$ such that
	 \begin{align} \label{eq:centroid}
		\frac{2(r_{j}^2 + r_j r_{j+1} +r_{j+1}^2)}{3(r_j + r_{j+1})} \sinc \left( \frac{\pi}{m}\right) = \abs{q_{j\ell}} =: \rho_j
\end{align}	 
while still $r_j \leq \rho_j \leq r_{j+1}$ for each $j$. To save a bit of notational effort, let us denote $\frac{\pi}{m}=\theta_0$.

Now, define the radii iteratively as follows:
\begin{align*}
	r_0 &=0 \\
	r_{j+1} &= \frac{1}{2}\left(\frac{3\rho_j}{2\sinc(\theta_0)}-r_j\right)\left(1+ \sqrt{1 + \frac{r_j}{\tfrac{3\rho_j}{2 \sinc(\theta_0)}-r_j}}\right)
\end{align*}
Elementary 	calculations then reveal that \eqref{eq:centroid} is satisfied. Defining the function $$S(t) = \frac{1}{2}\left(\frac{3}{2\sinc(\theta_0)}-t\right)\left(1+ \sqrt{1 + \frac{4t}{\tfrac{3}{2\sinc(\theta_0)}-t}}\right)$$ and $t_j = \tfrac{r_j}{\rho_j}$, we have
\begin{align*}
	r_{j+1} = \rho_j S(t_j) \Leftrightarrow t_{j+1} = \frac{\rho_j}{\rho_{j+1}}S(t_j).
\end{align*}
Let us collect a few properties of $S$ which will be important in the sequel.
\begin{lem} \label{lem:SProps}
	The function $S$ is decreasing, maps the interval $[0, \tfrac{3}{2\sinc(\theta_0)}]$ onto itself and obeys the inequality
	\begin{align*}
		t S(t) \leq \frac{1}{\sinc^2(\theta_0)} \left(1-  \frac{9}{4}\left( \sinc(\theta_0)t- 1\right)^2 \right)
	\end{align*}
\end{lem}
\begin{proof}
	Let us denote $\tfrac{3}{2\sinc(\theta_0)}=c$. We then have
	\begin{align*}
		S(t) &= \frac{1}{2}\left( (c-t) + \sqrt{(c-t)^2+ 4t(c-t)} \right) \\ &\Rightarrow S'(t) = \frac{1}{2} \left( \frac{2(c-t)-4t}{2\sqrt{(c-t)^2 + 4t(c-t)}}-1\right) < \frac{1}{2} \left( \frac{2(c-t)}{2\sqrt{(c-t)^2}}-1\right)=0
	\end{align*}
	for $t \in (0,c)$. Hence, $S$ is decreasing. This, together with continuity, in particular implies that $S([0,c]) = [S(c),S(0)] = [0, c]$. 
	
Finally, we have
\begin{align*}
	tS(t) \leq = \frac{c^2}{2} \cdot \frac{t}{c} \left( \left(1-\frac{t}{c}\right) + \sqrt{\left(1-\frac{t}{c}\right)^2+ 4\frac{t}{c}\left(1-\frac{t}{c}\right)} \right)
\end{align*}
Hence, if we denote $s(\theta) = \theta \left( (1-\theta) + \sqrt{(1-\theta)^2 + 4\theta(1-\theta)} \right)$, we need to prove
\begin{align*}
	s\left(\frac{2}{3} + \rho\right) \leq \frac{8}{9}  - 3 \rho^2 .
\end{align*}
We have
\begin{align*}
	s\left(\frac{2}{3} + \rho\right) &= \left(\frac{2}{3}+ \rho\right)\left( \left(\frac{1}{3}- \rho\right) + \sqrt{ \left(\frac{1}{3}- \rho\right)^2 + 4 \left(\frac{2}{3}+ \rho\right) \left(\frac{1}{3}- \rho\right)}\right) \\
	&= \left(\frac{2}{3}+ \rho\right)\left( \left(\frac{1}{3}- \rho\right) + \sqrt{ \frac{1}{9}- \frac{2}{3}\rho + \rho^2 +  \frac{8}{9} -\frac{4}{3}\rho - 4\rho^2} \right) \\
	&\leq \left(\frac{2}{3}+ \rho\right)\left( \frac{1}{3}- \rho + 1 -\rho- \frac{3}{2} \rho^2 \right) =\left(\frac{2}{3}+ \rho\right)\left( \frac{4}{3}- 2\rho - \frac{3}{2} \rho^2 \right) \\
	&\leq \frac{8}{9} -\frac{4}{3}\rho - \rho^2 + \frac{4}{3}\rho - 2 \rho^2 - \frac{3}{2}\rho^2 \leq \frac{8}{9} - 3\rho^2
\end{align*}
which was to be proven. We used the elementary inequality $\sqrt{1+t} \leq 1 + \tfrac{1}{2}t$

\end{proof}

We can now prove Lemma \ref{lem:radii}
\begin{proof}[Proof of Lemma \ref{lem:radii}]
We only need to prove that \eqref{eq:radii} is satisfied for the radii $r_i$ defined above, as long as $\rho_j \leq \Theta \cdot \rho$ for some constant $\Theta$. This is equivalent to proving that
\begin{align}
	S(t_j)\geq 1, \quad t_j \leq 1 \label{eq:t}
\end{align}
for the $j$ with $\rho_j \leq \Theta \cdot \rho$.

To prove \eqref{eq:t}, we use induction. For 
$j=0$, we simply need to note that $t_0 =0 \leq 1$ and $F(0)= \frac{3}{2\sinc(\theta_0)}\geq \frac{3}{2}>1$. The case $j=1$ is a bit trickier: 
	\begin{align*}
		t_1 = \frac{\rho_1}{\rho_2}\cdot\frac{3}{2\sinc(\theta_0)} = \frac{\sin\left(\frac{\pi}{m}\right)}{\sin\left(\frac{2\pi}{m}\right)} \frac{3}{2\sinc(\theta_0)} = \frac{1}{2\cos\left(\theta_0\right)} \frac{3}{2\sinc(\theta_0)} <1
	\end{align*}
	if $\theta_0 < 0.6378$, i.e. $m \geq \frac{\pi}{.6378} \approx 4.9267$, so $m \geq 5$ will suffice. For such $\theta_0$, we in particular have $\frac{3}{4\cos(\theta_0)}<1$, so the monotonicity of $S$ (Lemma \ref{lem:SProps}) implies
	\begin{align*}
		S(t_1) = S\left(\frac{3}{4\cos(\theta_0)\sinc(\theta_0)}\right)>S\left(\frac{1}{\sinc(\theta_0}\right) = \frac{1}{\sinc(\theta_0)}>1.
\end{align*}

	 Now for the induction step. Assume that $t_j<1$ and $S(t_j)>1$.  The monotonicity of $S$ then implies, since $t_j <1 < \sinc(\theta_0)^{-1}$
	 \begin{align*}
	 	S(t_j)>S(1)> S\left(\frac{1}{\sinc(\theta_0)}\right)= \frac{1}{\sinc(\theta_0)}
	 \end{align*}
	 This on the other hand implies
	 \begin{align*}
	 	S(t_{j+1}) &= S\left( \frac{\rho_j}{\rho_{j+1}}F(t_j)\right)< S\left( \frac{1}{\sinc(\theta_0)} \frac{\rho_j}{\rho_{j+1}}\right) \\
	 	 \Rightarrow t_{j+2} &= \frac{\rho_{j+1}}{\rho_{j+2}} S(t_{j+1}) <  \frac{\rho_{j+1}}{\rho_{j+2}} F\left( \frac{1}{\sinc(\theta_0)} \frac{\rho_j}{\rho_{j+1}}\right) \\
	 	&= \sinc(\theta_0)\frac{\rho_{j+1}^2}{\rho_{j+2}\rho_j} \frac{\rho_j}{\sinc(\theta_0)\rho_{j+1}}F\left( \frac{1}{\sinc(\theta_0)} \frac{\rho_j}{\rho_{j+1}}\right) \\
	 	&\leq\sinc(\theta_0)\frac{\rho_{j+1}^2}{\rho_{j+2}\rho_j} \cdot \frac{1}{\sinc(\theta_0)^2} \left( 1-\frac{9}{4}\left(\frac{\rho_j}{\rho_{j+1}}-1\right)^2 \right),
	 \end{align*}
	 where we in the last step applied Lemma \ref{lem:SProps}. Now let us note that the addition formulas for $\sin$ and $\cos$ imply
	 \begin{align*}
	 	\frac{\rho_j}{\rho_{j+1}} = \frac{1}{\cos(\theta_0) + \sin(\theta_0) \cot(\tfrac{\pi j}{m})}, \ \frac{\rho_{j+1}^2}{\rho_{j+2}\rho_j} = \frac{(\cos(\theta_0) +\cot(\tfrac{\pi j}{m})\sin(\theta_0))^2}{\cos(2\theta_0) + \sin(2\theta_0) \cot(\tfrac{\pi j}{m})}.
	 \end{align*}
	 Hence, in order for $t_{j+2}$ to be smaller than $1$, we need
	 \begin{align*}
	 	\frac{1}{\sinc(\theta_0)} \cdot \frac{(\cos(\theta_0) +\cot(\tfrac{\pi j}{m})\sin(\theta_0))^2}{\cos(2\theta_0) + \sin(2\theta_0) \cot(\tfrac{\pi j}{m})} \cdot \left(1- \frac{9}{4}\left(\frac{1}{\cos(\theta_0) + \sin(\theta_0) \cot(\tfrac{\pi j}{m})} -1\right)^2 \right) <1
	 \end{align*}
	 This inequality can be transformed into a quadratic inequality in $\cot(\tfrac{\pi j}{m})$:
	 \begin{align*}
	 	\cot^2(\tfrac{\pi j}{m}) &+ \frac{4}{5}\cot(\tfrac{\pi j}{m})\left( \frac{\sin(2\theta_0)\sinc(\theta_0)}{\sin^2(\theta_0)}+ \frac{10}{4}\frac{\cos(\theta_0)}{\sin(\theta_0)} - \frac{18}{4\sin(\theta_0)}\right) \\
	 	&+ \frac{4}{5}\left(\frac{\cos(2\theta_0)\sinc(\theta_0)}{\sin^2(\theta_0)} + \frac{5}{4}\frac{\cos^2(\theta_0)}{\sin^2(\theta_0)}+\frac{9}{4\sin^2(\theta_0)}- \frac{18\cos(\theta_0)}{4\sin^2(\theta_0)}\right)>0 
	 \end{align*}
	 Let us give the coefficent before $\cot(\frac{\pi j}{m})$ the name $p(\theta_0)$, and the constant coefficient the name $q(\theta_0)$.  $p$ is decreasing and $q$ is increasing for  growing $\theta$. Hence, if the  inequality is satisfied for a certain $\theta_0$, it will be satisfied for smaller values of $\theta_0$ as well. Since $m\geq 5$, $\theta_0 \geq \frac{\pi}{5}$. Hence, we may assume $\theta_0 = \frac{\pi}{5}$. For this value of $\theta_0$, the inequality above reads  $\cot^2(\tfrac{\pi j}{m}) - \cot(\tfrac{\pi j}{m}) \cdot .7704 -.687285>0$, which is satisfied for $ \cot(\frac{\pi j}{m}) >1.3$, which corresponds to $\sin(\frac{\pi j}{m})<.60972$, i.e
	 \begin{align*}
	 	\rho_j \leq 1.21942\cdot\rho.
	 \end{align*}
	The proof is finished.

	\end{proof}

\end{document}